\documentclass[hidelinks,onefignum,onetabnum]{siamart220329}

\usepackage[T1]{fontenc}
\usepackage{amsmath}
\usepackage{amssymb}
\usepackage{tikz}

\usetikzlibrary{decorations.text}

\newcommand{\R}{\mathbb{R}}
\newcommand{\N}{\mathbb{N}}
\newcommand{\Q}{\mathbb{Q}}
\newcommand{\ex}{\operatorname{ex}}
\newcommand{\minex}{\operatorname{minex}}
\newcommand{\problem}{\textsc{Multi-Sink }\allowbreak\mbox{\ensuremath{\alpha}\textsc{-Commodity }}\allowbreak\textsc{Flow}}
\newcommand{\optset}{S^{\ast}}
\newcommand{\greedyset}{S^{\textrm{G}}}
\newcommand{\nagreedyset}{\tilde{S}^{\textrm{G}}}
\newcommand{\U}{U}
\newcommand{\I}{\mathcal{I}}
\newcommand{\lastk}{\bar{k}}
\newcommand{\oneto}[1]{[#1]}
\newcommand{\setf}{\mathcal{F}}
\newcommand{\e}{\mathrm{e}}
\newcommand{\srfuncs}{\tilde{\setf}}
\newcommand{\augfuncs}{\setf}
\newcommand{\rqfuncs}{\setf}
\newcommand{\gafuncs}{\tilde{\setf}}
\renewcommand{\vec}[1]{\boldsymbol{#1}}

\newsiamremark{rem}{Remark}
\newsiamremark{hypothesis}{Hypothesis}
\crefname{hypothesis}{Hypothesis}{Hypotheses}
\newsiamthm{claim}{Claim}
\newsiamthm{prop}{Proposition}
\newsiamthm{thm}{Theorem}
\newsiamthm{defn}{Definition}
\newsiamthm{cor}{Corollary}
\newsiamthm{lem}{Lemma}
\newsiamthm{example}{Example}

\headers{Unified Greedy Approximability Beyond Submodular Maximization}{Y. Disser, D. Weckbecker}

\title{Unified Greedy Approximability Beyond Submodular Maximization\thanks{An extended abstract of this paper appeared in \cite{DisserWeckbecker/20}.\funding{Supported by DFG grant DI 2041/2.}}}

\author{Yann Disser\thanks{TU Darmstadt
		(\email{\{disser$|$weckbecker\}@mathematik.tu-darmstadt.de}).}
	\and David Weckbecker\footnotemark[2]}

\usepackage{amsopn}

\ifpdf
\hypersetup{
  pdftitle={Unified Greedy Approximability Beyond Submodular Maximization},
  pdfauthor={Yann Disser and David Weckbecker}
}
\fi

\begin{document}

\maketitle

\begin{abstract}
We consider classes of objective functions of cardinality-constrained
maximization problems for which the greedy algorithm guarantees a
constant approximation. We propose the new class of $\gamma$-$\alpha$-augmentable
functions and prove that it encompasses several important subclasses,
such as functions of bounded submodularity ratio, $\alpha$-augmentable
functions, and weighted rank functions of an independence system of
bounded rank quotient -- as well as additional objective functions
for which the greedy algorithm yields an approximation. For this general
class of functions, we show a tight bound of~$\frac{\alpha}{\gamma}\cdot\frac{\mathrm{e}^{\alpha}}{\mathrm{e}^{\alpha}-1}$
on the approximation ratio of the greedy algorithm that tightly interpolates
between bounds from the literature for functions of bounded submodularity
ratio and for $\alpha$-augmentable functions. In particular, as a
by-product, we close a gap in {[}Math.Prog., 2020{]} by obtaining
a tight lower bound for $\alpha$-augmentable functions for all $\alpha\geq1$.
For weighted rank functions of independence systems, our tight bound
becomes~$\frac{\alpha}{\gamma}$, which recovers the known bound
of $1/q$ for independence systems of rank quotient at least~$q$.
\end{abstract}



\section{Introduction}

We consider cardinality-constrained maximization problems of the form
\begin{align*}
	\max & \,\,f(X)\\
	\mathrm{s.t.} & \,\,|X|\leq k\\
	& \,\,X\subseteq\U,
\end{align*}
with a \textit{monotone} objective function $f\colon2^{\U}\to\mathbb{R}_{\geq0}$
over a finite ground set $U$. Additional constraints of the form
$X\in\mathcal{X}$ can be modeled by the monotone objective $f'(X):=\max\{f(Y)|Y\in2^{X}\cap\mathcal{X}\}$.
In this way, every combinatorial, cardinality-constrained maximization
problem with monotone objective can be captured, and we adopt this
framework throughout the paper.\footnote{Note that the objective function $f$ may be computationally hard
	to evaluate. If we assume that the greedy algorithm has oracle access
	to $f$, it requires $O(|U|k)$ queries to the oracle.} For example, the maximum weighted matching problem on a graph $G=(V,E)$
with edge weights~$w\colon E\to\mathbb{R}_{\geq0}$ yields the objective
function $\mbox{\ensuremath{f(X\!\subseteq\!E)=\max\{\sum_{e\in M}w(e)|M\subseteq X,M\text{ is a matching in }G\}}}$.

We focus on the performance of the\textit{ greedy algorithm}. This algorithm
iteratively produces a solution $\mbox{\ensuremath{\greedyset_{f,k}=\{x_{1},\dots,x_{k}\}}}$
with 
\[
x_{i}\in{\textstyle \arg\max_{x\in U\setminus\{x_{1},\dots,x_{i-1}\}}}f(\{x_{1},\dots,x_{i-1}\}\cup\{x\}),
\]
for all $i\in\oneto k:=\{1,\dots,k\}$, i.e., it adds elements such
that the increase in objective value is maximized in each step. The
greedy algorithm is inherently incremental and may be regarded as
the most natural approach for incrementally building up infrastructures
that support changing active solutions (in the sense of the definition
$f'(X)$ above). While this algorithm is widely used in practical
applications, greedy solutions can be arbitrarily far away from optimal
(e.g., for the knapsack problem). A natural question in this context
is, for which objective functions $f$ the greedy algorithm gives
a good solution. We are interested in characterizing these objective
functions.

Note that we consider the \textit{adaptive }greedy solution $\smash{\greedyset_{f,k}}$
as opposed to the \textit{non-adaptive} greedy solution $\smash{\nagreedyset_{f,k}:=\greedyset_{f,\min\{k,\bar{k}\}}}$,
where $\lastk\in\oneto{|U|}$ is the smallest cardinality such that
$\mbox{\ensuremath{f(\greedyset_{\lastk}\cup\{x\})=f(\greedyset_{\lastk})}}$
for all $\mbox{\ensuremath{x\in U\setminus\greedyset_{\lastk}}}$.
In other words, the non-adaptive greedy algorithm terminates as soon
as it cannot improve the solution further. This non-adaptive variant
of the greedy algorithm has often been considered in the early literature~(e.g.,
\cite{Jenkyns1976,Korte1978,Nemhauser1978,Nemhauser1978II}). Note,
that for submodular functions, i.e., functions with $\mbox{\ensuremath{f(X\cup Y)+f(X\cap Y)\leq f(X)+f(Y)}}$
for all $X,Y\subseteq U$, there is no difference between these two
variants, and for our purposes both variants are interchangeable in
the following sense.

Formally, we measure the quality of the greedy algorithm on a set
of objectives~$\setf$ by the approximation ratio $\sup_{f\in\setf}\max_{k\in[|U_{f}|]}\!f(\optset_{f,k})/f(X_{f,k})$,
where $U_{f}$ is the ground set of the function $f\in\mathcal{F}$,
$\optset_{f,k}\!\in\!\arg\max_{X\subseteq U:|X|\leq k}f(X)$ denotes
an optimum solution of cardinality at most~$k$, and $\smash{X_{f,k}\in\{\greedyset_{f,k},\nagreedyset_{f,k}\}}$
refers to the \mbox{(non-)adaptive} greedy solution of cardinality~$k$.
We claim that the approximation ratios of both variants of the greedy
algorithm coincide. To see this, observe that the non-adaptive setting
is more restrictive, and that every lower bound instance in the non-adaptive
setting can be made adaptive by introducing additional elements that
add a vanishingly small but positive objective value when added to
every solution. This implies that all our bounds on the approximation
ratio of the (adaptive) greedy algorithm immediately apply to both
variants.

From now on, we write $\greedyset_{k}:=\greedyset_{f,k}$ and $\optset_{k}:=\optset_{f,k}$,
whenever $f$ is clear from the context. In these terms, we are interested
in characterizing the set of objectives for which the greedy algorithm
has a bounded approximation ratio. Known examples include the objectives
of maximum (weighted) $\mbox{(b-)matching}$, maximum (weighted) coverage,
and many more~\cite{BernsteinDisserGrossHimburg/20,Bian2017,Das2018,Krause2011,Williamson2011},
and we additionally introduce a multi-commodity flow problem (Section~\ref{sec:lower_bound_alph-aug=000026separation}),
where the greedy algorithm yields an approximation.

A well-known class of functions for which the greedy algorithm has
a bounded approximation ratio of (exactly) $\frac{\e}{\e-1}$ are
the monotone, submodular functions~\cite{Nemhauser1978}. This class
includes the maximum coverage problem, but fails to capture many other
greedily approximable settings. See Figure~\ref{fig:classes} along
with the following.

Das and Kempe~\cite{Das2018} introduced the class of functions of
bounded \textit{submodularity ratio }as a generalization of submodular
functions. Importantly, its definition depends on the greedy solutions
for different cardinalities. We adapt and weaken the definition from~\cite{Das2018}
for consistency, by restricting ourselves to greedy solutions and
by minimizing over all cardinalities.
\begin{defn}
	[\cite{Das2018}]The \textit{weak submodularity ratio }of $f\colon2^{\U}\rightarrow\R_{\geq0}$
	is (using $\frac{0}{0}:=1$)\label{def:submodularity_ratio}
	\[
	\gamma(f):=\min_{X\in\{\greedyset_{0},\dots,\greedyset_{\lastk}\},Y\subseteq\U\setminus X}\frac{\sum_{y\in Y}(f(X\cup\{y\})-f(X))}{f(X\cup Y)-f(X)}\in[0,1].
	\]
\end{defn}

Das and Kempe~\cite{Das2018} showed an upper bound of $\frac{\e^{\gamma}}{\e^{\gamma}-1}$
on the approximation ratio of the greedy algorithm for the set of
all monotone functions with submodularity ratio at least $\gamma>0$,
and Bian et al.~\cite{Bian2017} extended this to a tight bound that
is additionally parameterized by the curvature of the objective. Since
submodular functions have submodularity ratio $1$, this bound generalizes
the submodular bound. Crucially, it is easy to verify that these
results carry over to the set~$\gafuncs_{\gamma}$ of all monotone
functions with \textit{weak} submodularity ratio at least $\gamma>0$.\footnote{Here and throughout we use the notation $\tilde{\mathcal{F}}$ as
	opposed to $\mathcal{F}$ to refer to a function class based on a
	\textit{weak} definition.}

Another generalization of submodularity was proposed by Bernstein
et al.~\cite{BernsteinDisserGrossHimburg/20}. We extend the definition
by a weakened variant in order to bring it more in line with Definition~\ref{def:submodularity_ratio}.
\begin{defn}
	[\cite{BernsteinDisserGrossHimburg/20}]The function $f\colon2^{\U}\rightarrow\R_{\geq0}$
	is \textit{(weakly) $\alpha$-augmentable} for $\alpha\geq1$, if, for
	every $X\subseteq U$ ($X\in\{\greedyset_{0},\dots,\greedyset_{\lastk}\}$)
	and $Y\subseteq\U$ with $Y\nsubseteq X$, there exists an element
	$y\in Y\setminus X$ with
	\[
	f(X\cup\{y\})-f(X)\geq\frac{f(X\cup Y)-\alpha f(X)}{\left|Y\right|}.
	\]
\end{defn}

Bernstein et al.~showed that the greedy algorithm has an approximation
ratio of at most $\alpha\cdot\frac{\e^{\alpha}}{\e^{\alpha}-1}$ on
the set~$\augfuncs_{\alpha}$ of monotone, $\alpha$-augmentable
functions, for $\alpha\geq1$, and that this bound is tight for $\alpha\in\{1,2\}$
and in the limit~$\alpha\to\infty$. Since submodular functions are
$1$-augmentable, this bound again generalizes the submodular bound.
The class of $\alpha$-augmentable problems captures the objective
of the maximum (weighted) $\alpha$-dimensional matching problem,
which is not submodular. In this paper, we introduce a natural $\alpha$-commodity
flow variant that is $\alpha$-augmentable, and we prove a tight lower
bound on the approximation ratio for all $\alpha\geq1$.

Another well-known setting, besides submodularity, where the greedy
algorithm has a bounded approximation ratio, are weighted rank functions
of independence systems of bounded rank quotient~\cite{KorteVygen2012}.
An \textit{independence system} is a tuple $(\U,\I\subseteq2^{\U})$,
where $\I$ is closed under taking subsets and $\emptyset\in\I$.
For a given weight\linebreak function \mbox{$w\colon\U\to\mathbb{R}_{\geq0}$}, the
\textit{weighted rank function} of $(\U,\I)$ is given by\linebreak \mbox{$f(X)=\max\{\sum_{x\in Y}w(x)|Y\in\I\cap2^{X}\}$}.
The \textit{rank quotient }of an independence system $(\U,\I)$ is $q(\U,\I):=\min_{X\subseteq\U}\min_{B,B'\in\mathcal{B}(X)}|B|/|B'|$,
where we set~\mbox{$\frac{0}{0}:=1$}, and the set~$\mathcal{B}(X)$
of all \textit{bases }of some set $X\subseteq\U$ is defined to be the
set of inclusion-wise maximal subsets of $\I\cap2^{X}$, i.e., \mbox{$\mathcal{B}(X):=\{B\in\I\cap2^{X}|\forall x\in X\setminus B\colon B\cup\{x\}\notin\I\}$}.
Jenkyns~\cite{Jenkyns1976} and Korte and Hausmann~\cite{Korte1978}
showed that the greedy algorithm has an approximation ratio of exactly
$1/q$ on the set~$\rqfuncs_{q}$ of all weighted rank functions
of independence systems with rank quotient at least $q>0$.\footnote{Note that we abuse notation, since, e.g., $\rqfuncs_{\alpha}\neq\rqfuncs_{q}$
	for $\alpha=q=1$. However, the set of functions we are referring
	to will always be clear by the naming of the indices.}

\paragraph{Our results.}

Our goal is to unify and to generalize the above classes of functions
on which the greedy algorithm has a bounded approximation ratio. To
this end, we first observe that each one of the classes $\srfuncs_{\gamma}$,
$\augfuncs_{\alpha}$, and~$\rqfuncs_{q}$ uniquely captures greedily
approximable objectives (cf.~Figure~\ref{fig:classes} and Propositions~\ref{prop:separating_alpha_aug_lite},
\ref{prop:separating_submod_ratio}, \ref{prop:separating_rank_quot},
and \ref{prop:separating_alpha_aug}). In particular, we construct
a natural $\alpha$-augmentable variant of multi-commodity flow that
does not have bounded (weak) submodularity ratio (for~$\alpha\in\mathbb{N}\setminus\{1\}$)
and cannot be expressed as the maximization of a weighted rank function.
Besides the $\alpha$-dimensional matching problem, to our knowledge,
the problem introduced in Section~\ref{sec:lower_bound_alph-aug=000026separation}
is the only other natural $\alpha$-augmentable problem to date.
\begin{prop}
	For every $\gamma,q\in(0,1)$ and $\alpha\geq1$, it holds that\label{prop:separation}
	\[
	\srfuncs_{\gamma}\nsubseteq(\augfuncs_{\alpha}\cup\rqfuncs_{q})\quad\mathrm{and}\quad\augfuncs_{\alpha}\nsubseteq(\srfuncs_{\gamma}\cup\rqfuncs_{q})\quad\mathrm{and}\quad\rqfuncs_{q}\nsubseteq(\srfuncs_{\gamma}\cup\augfuncs_{\alpha}).
	\]
\end{prop}

This motivates the following definition to consolidate all three classes.

\begin{defn}
	The function $f\colon2^{\U}\rightarrow\R_{\geq0}$ is (weakly) $\gamma$-$\alpha$-\textit{augmentable}
	for $\gamma\in(0,1]$ and $\alpha\geq\gamma$ if, for all sets $X\subseteq\U$
	($X\in\{\greedyset_{0},...,\greedyset_{\lastk}\}$) and all $Y\subseteq\U$
	with $Y\nsubseteq X$, there exists $y\in Y$ with
	\[
	f(X\cup\{y\})-f(X)\geq\frac{\gamma f(X\cup Y)-\alpha f(X)}{|Y|}.
	\]
\end{defn}

Note that we need to consider the weak variant of this definition
if we hope to encompass the class $\srfuncs_{\gamma}$, which enforces
its defining property only for ``greedy sets'', however, any upper
bound on the approximation ratio immediately carries over to the same
bound in the stronger definition. Also note that $\gamma$-$\alpha$-augmentability
only requires $\alpha\geq\gamma$, unlike $\alpha$-augmentability
where $\alpha\geq1$. This is in line with the definitions of $\alpha$-augmentability
where $\gamma=1$ and of the submodularity ratio where $\alpha=\gamma$.
We let $\gafuncs_{\gamma,\alpha}$ denote the set of all weakly $\gamma$-$\alpha$-augmentable
functions. The first part of our main result is that this set encompasses
all functions in $\srfuncs_{\gamma}\cup\augfuncs_{\alpha}\cup\rqfuncs_{q}$
and captures additional functions (cf.~Figure~\ref{fig:classes}).
Formally, we show the following (cf.~Propositions~\ref{prop:gamma-alpha_unifies_alpha_and_submod-ratio}
and~\ref{prop:F_separates_gam-alph-aug}).
\begin{thm}
	\label{thm:gam-alph-aug_weaker_than_rest}For every $\gamma,q\in(0,1]$,
	every $\gamma'\in(0,1)$, every $\alpha\geq1$, and every $\alpha'\geq\gamma'$,
	it holds that
	\[
	\gafuncs_{\gamma,\max\{\alpha,1/q\}}\supseteq\srfuncs_{\gamma}\cup\augfuncs_{\alpha}\cup\rqfuncs_{q}\qquad\mathrm{and}\qquad\gafuncs_{\gamma',\alpha'}\nsubseteq\srfuncs_{\gamma}\cup\augfuncs_{\alpha}\cup\rqfuncs_{q}.
	\]
\end{thm}

Note that $\alpha'$ and $\gamma'$ in Theorem \ref{thm:gam-alph-aug_weaker_than_rest}
do not depend on $\alpha$, $\gamma$ and $q$. The second part of
our main result is a tight bound on the approximation ratio of the
greedy algorithm on $\gafuncs_{\gamma,\alpha}$ (cf.~Theorems~\ref{thm:gam-alph-aug_upper_bound} and~\ref{thm:general_lower_bound}).
\begin{thm}
	The approximation ratio of the greedy algorithm on the class $\gafuncs_{\gamma,\alpha}$
	of monotone, weakly $\gamma$-$\alpha$-augmentable functions, with
	$\gamma\in(0,1]$ and $\alpha\geq\gamma$, is exactly\label{thm:tight_bound}
	\[
	\frac{\alpha}{\gamma}\cdot\frac{\mathrm{e}^{\alpha}}{\mathrm{e}^{\alpha}-1}.
	\]
\end{thm}

Importantly, this bound recovers exactly the known bound for functions
of\linebreak bounded submodularity ratio, since~\mbox{$\srfuncs_{\gamma}\subseteq\gafuncs_{\gamma,\gamma}$},
as well as the known bound for\linebreak \mbox{$\alpha$-augmentable} functions, since~$\augfuncs_{\alpha}\subseteq\gafuncs_{1,\alpha}$.
In that sense, our new bound interpolates tightly between these two
bounds and generalizes them. In addition, our tight lower bound for~$\gafuncs_{1,\alpha}$
is obtained with an $\alpha$-augmentable function. This means that,
in particular, we are able to close the gap left in~\cite{BernsteinDisserGrossHimburg/20},
by showing a tight lower bound for $\alpha$-augmentable objectives,
for all~$\alpha\geq1$ (cf.~Propositions~\ref{prop:approx-ratio_of_F}
and~\ref{prop:F_1,alph,k_is_alph-aug}).
\begin{cor}
	The approximation ratio of the greedy algorithm on the class~$\augfuncs_{\alpha}$
	of monotone, $\alpha$-augmentable functions is exactly $\alpha\cdot\frac{\e^{\alpha}}{\e^{\alpha}-1}$
	for all $\alpha\geq1$.\label{cor:augmentable_lower_bound}
\end{cor}

Finally, we are also able to show a tight bound of $\alpha/\gamma$
for $\gamma$-$\alpha$-augmentable, weighted rank functions on independence
systems (cf.~Propositions~\ref{prop:upper-bound_gamma-alpha-aug_ind-sys}
and~\ref{prop:lower-bound_gamma-alpha-aug_ind-sys}). Since $\rqfuncs_{q}\subseteq\gafuncs_{1,1/q}$
(by~Theorem~\ref{thm:gam-alph-aug_weaker_than_rest}), our bound
recovers exactly the known bound of $1/q$ for the approximation ratio
of the greedy algorithm when the rank quotient is bounded from below
by $q>0$. This means that the class of monotone, weakly $\gamma$-$\alpha$-augmentable
functions truly unifies and generalizes the three classes $\srfuncs_{\gamma}$,
$\augfuncs_{\alpha}$, and $\rqfuncs_{q}$ of greedily approximable
functions (cf.~Figure~\ref{fig:classes}). Note that, in particular,
the lower bound is tight already for $\alpha$-augmentable functions,
which implies a tight bound of $\alpha$ for the approximation ratio
of the greedy algorithm on $\alpha$-augmentable weighted rank functions.
\begin{thm}
	Let $\rqfuncs_{\mathrm{IS}}:=\bigcup_{q\in(0,1]}\rqfuncs_{q}$ be
	the set of weighted rank functions on some independence system. The
	approximation ratio of the greedy algorithm on the class $\gafuncs_{\gamma,\alpha}\cap\rqfuncs_{\mathrm{IS}}$,
	with $\gamma\in(0,1]$ and $\alpha\geq\gamma$, is exactly~$\frac{\alpha}{\gamma}$.\label{thm:IS_bound}
\end{thm}

\begin{figure}
	\begin{centering}
		\begin{center}
			\scalebox{.8}{
				\begin{tikzpicture}
					\pgfmathsetmacro{\ls}{-2pt}
					\draw[thick,postaction={decorate,decoration={text along path,raise=3pt,pre length=7.3cm,text={|\small|greedily approximable}}}] (-.1,0) arc(180:-180:7.45cm and 4.7cm); 	
					\draw[thick,postaction={decorate,decoration={text along path,raise=2pt,pre length=.8cm,text={|\small|submodular}}}] (3.8,.2) arc(180:-180:1.2cm and 1.2cm); 	
					\draw[thick,postaction={decorate,decoration={text along path,raise=2pt,pre length=4.2cm,text={|\small|1-augmentable}}}] (3.45,.4) arc(180:-180:3.1cm and 1.8cm); 	
					\draw[thick,postaction={decorate,decoration={text along path,raise=2pt,pre length=4.6cm,text={|\small|2-augmentable}}}] (3.25,.4) arc(180:-180:4.225cm and 2.3cm); 	
					\draw[thick,postaction={decorate,decoration={text along path,raise=2pt,pre length=5.3cm,text={|\small|{$\alpha$}-augmentable}}}] (3.05,.5) arc(180:-180:5.45cm and 3cm); 	
					\draw[thick,postaction={decorate,decoration={text along path,raise=-9pt,pre length=4.4cm,text={|\small|rank quotient {$q$}}}}] (2.4,-1.55) arc(-180:180:4.9cm and 2cm); 	
					\draw[thick,postaction={decorate,decoration={text along path,raise=3pt,pre length=0cm,text={|\small|weak submod. ratio {$\gamma$}}}}] (1.5,.65) arc(180:-180:3.45cm and 1.8cm); 
					\draw[very thick,dashed,red,postaction={decorate,decoration={text along path,raise=3pt,pre length=7.1cm,text={|\small\color{red}|weakly {$\gamma$}-{$\alpha$}-augmentable}}}] (.7,0) arc(180:-180:6.8cm and 4.2cm); 
					
					\node[fill,scale=.8,label={\footnotesize$\textsc{Knapsack}$}] at (1,3) {};
					\node[fill,scale=.8,label={\footnotesize$\textsc{Disjoint Paths}$}] at (2,4) {}; 	\node[fill,scale=.8,label={[align=center]\footnotesize$\textsc{Independent}$ \\[\ls] \footnotesize$\textsc{Set}$}] at (12.6,3.6) {};
					\node[fill,scale=.8,label={\footnotesize$\textsc{Coverage}$}] at (5,.47) {};
					\node[fill,scale=.8,label={[align=center]\footnotesize$1\textsc{-Dim}$ \\[\ls] \footnotesize$\textsc{Matching}$}] at (5,-.7) {};
					\node[fill,circle,red,scale=.7,label={[align=center,red]\footnotesize$\textsc{Multi-Sink}$ \\[\ls] \footnotesize$1\textsc{-Commodity}$ \\[\ls] \footnotesize$\textsc{Flow}$}] at (6.9,.6) {};
					\node[fill,scale=.8,label={[align=center]\footnotesize$\textsc{Bipartite}$ \\[\ls] \footnotesize$\textsc{Matching}$}] at (10.15,-1.1) {};
					\node[fill,scale=.8,label={[align=center]\footnotesize$\textsc{Bridge-}$ \\[\ls] \footnotesize$\textsc{Flow}$}] at (10.55,.5) {};
					\node[fill,circle,red,scale=.7,label={[align=center,red]\footnotesize$\textsc{Multi-Sink}$ \\[\ls] \footnotesize$\alpha\textsc{-Commodity}$ \\[\ls] \footnotesize$\textsc{Flow}$}] at (12.75,.3) {};
					\node[fill,scale=.8,label={[align=center]\footnotesize$\alpha\textsc{-Dim}$ \\[\ls] \footnotesize$\textsc{Matching}$}] at (11.3,-1.7) {};
					\node[fill,circle,red,scale=.7,label={[red]\footnotesize$ F_{1,1,k'} $}] at (8.9,.6) {};
					\node[fill,circle,red,scale=.7,label={[red]\footnotesize$ f^q $}] at (5,-2.8) {};
					\node[fill,circle,red,scale=.7,label={[red]\footnotesize$ F_{1,\alpha,k'} $}] at (12.1,1.8) {};
					\node[fill,circle,red,scale=.7,label={[red]\footnotesize$ f^\gamma $}] at (2.3,.4) {};
					\node[fill,circle,red,scale=.7,label={[red]\footnotesize$ f^{\gamma,\alpha} $}] at (.35,0) {};
					\node[fill,circle,red,scale=.7,label={[red]\footnotesize$ F_{\gamma,\alpha,k'} $}] at (4.3,2.9) {};
			\end{tikzpicture}}
		\end{center}
		\par\end{centering}
	\caption{Relation of the different problem classes. Newly introduced classes
		and problems are marked in red with dashed lines and round nodes. The parameter $k'$ is chosen sufficiently
		large, depending on $\gamma$ and~$\alpha$.\label{fig:classes}}
\end{figure}
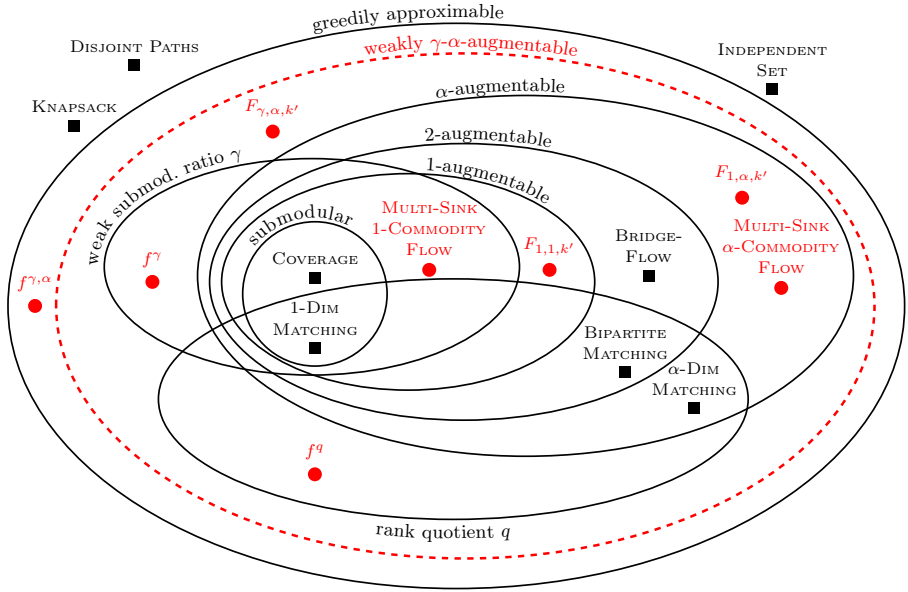

\paragraph{Related Work.}

We can view our cardinality-constrained maximization framework as
a special case of maximization over an independence system. In particular,
the cardinality-constraint can be expressed as a uniform matroid constraint
\cite{KorteVygen2012}. From that perspective, the most basic, non-trivial
setting is the maximization of a linear~(i.e., modular) objective
over an independence system. Regarding the approximation ratio of
the greedy algorithm, this classic setting is equivalent to the maximization
of a weighted rank function, as considered in Theorem~\ref{thm:IS_bound}.
This is easy to see by considering the non-adaptive variant of the
greedy algorithm, and by observing that the greedy solution is guaranteed
to remain feasible while the algorithm makes progress~(cf.~Lemma~\ref{lem:adding_element_to_greedy_set}).

In that sense, the perfomance of the greedy algorithm for weighted
rank function maximization has extensively been studied in the past.
Rado~\cite{Rado1942} showed that the greedy algorithm is optimal
for all weight functions if the underlying independence system is
a matroid, and Edmonds~\cite{Edmonds1971} established the reverse
implication. Jenkyns~\cite{Jenkyns1976} extended this result by
showing an upper bound of~$1/q$~for the approximation ratio of
the greedy algorithm on independence systems with rank quotient $q$,
and Korte and Hausmann~\cite{Korte1978} gave a tight lower bound.
Years later, Mestre~\cite{Mestre2006} independently proved this
tight bound for the subclass of $k$-extendible independence systems.
Bouchet~\cite{Bouchet1987} gave a different generalization of the
result by Rado and Edmonds by showing that the greedy algorithm remains
optimal on symmetrical matroids.

Another prominent setting is the maximization of a submodular function
over an independence system. Again, this includes cardinality-constrained
maximization of a submodular objective, which is equivalent to submodular
maximization over a uniform matroid. Nemhauser, Wolsey, and Fisher~\cite{Nemhauser1978II}
showed that the greedy algorithm has a tight approximation ratio of
$\frac{\e}{\e-1}$ for maximizing a monotone, submodular function
under a cardinality-constraint. Krause et al.~\cite{Krause2008}
observed that the approximation ratio is unbounded when maximizing
the minimum of two monotone, submodular functions. Non-monotone submodular
maximization over a cardinality-constraint (and knapsack constraints)
was considered by Lee et al.~\cite{Lee2009}. Feldman et al.~\cite{Feldman2011}
analyzed a variant of the continuous greedy algorithm~\cite{Vondrak2008}
and showed an upper bound on its approximation ratio of $(1/e-o(1))^{-1}$.
This bound for the non-monotone case with cardinality-constraint was
later improved by Buchbinder et al.~\cite{Buchbinder2014} and Ene
and Nguyen~\cite{Ene2016} by further adapting the (continuous) greedy
algorithm. For maximizing a submodular function subject to $k$-extentible
system and $k$-systems constraints, Feldman et al.~\cite{FeldmanHarshawKarbasi/17,FeldmanHarshawKarbasi/20}
considered three variants of the greedy algorithm, a repeated greedy,
a sample greedy and a simultaneous greedy. They were able to show
approximation ratios of $k+O(1)$ for $k$-extendible system constraints
and $k+O(\sqrt{k})$ for $k$-system constraints.

Maximization of a monotone, submodular function over a matroid was
considered by Vondrák~\cite{Vondrak2008} and by Calinescu et al.~\cite{Calinescu2011},
who showed that the continuous greedy algorithm has an approximation
ratio of~$\frac{\e}{\e-1}$ in this setting. Nemhauser, Wolsey, and
Fisher~\cite{Nemhauser1978II}, showed an upper bound of $p+1$ for
the regular greedy algorithm when maximizing over the intersection
of $p$ matroids. A generalization of this upper bound to the setting
of maximizing subject to a $p$-system constraint was later proven
by Calinescu et al.~\cite{Calinescu2011}. Conforti and Cornuejols~\cite{Conforti1984}
gave an upper bound of $p+c$ depending on the curvature $c$ of the
monotone submodular function -- this interpolates between the submodular
bound of~\cite{Nemhauser1978II} ($c=1$) and the linear bound of~\cite{Korte1978}
($c=0$). Vondrák~\cite{Vondrak2010} showed that the continuous
greedy algorithm has an approximation ratio of at most $c\frac{\e^{c}}{\e^{c}-1}$
over an arbitrary matroid, and Sviridenko, Vondrák, and Ward~\cite{Sviridenkko2015}
showed an improved upper bound of $\frac{\e}{\e-c}$ for the approximation
ratio of a modified continuous greedy algorithm over a uniform matroid
(i.e., a cardinality-constraint).

Other variants of the problem setting include the maximization of
a monotone, submodular function over a knapsack constraint~\cite{Sviridenko2004},
and robust submodular maximization~\cite{Anari2019,Orlin2015}.

\section{Weak Submodularity Ratio, $\alpha$-Augmentability, and Independence
	Systems\label{sec:lower_bound_alph-aug=000026separation}}

In this section, we prove Proposition~\ref{prop:separation}, i.e.,
we separate the function classes $\srfuncs_{\gamma}$, $\augfuncs_{\alpha}$,
and $\rqfuncs_{q}$. We start by introducing a natural $\alpha$-commodity
flow problem that models, e.g., production processes where output
is limited by availability of all components. The objective of this
problem is (exactly) $\alpha$-augmentable, but, for $\alpha\in\mathbb{N}\setminus\{1\}$,
does not have a bounded (weak) submodularity ratio and cannot be expressed
as a weighted rank function over an independence system. This problem
also gives a tight lower bound for the approximation ratio of the
greedy algorithm on $\alpha$-augmentable functions, for $\alpha\in\mathbb{N}$.
We will extend this lower bound to all $\alpha\geq1$ in Section~\ref{subsec:lower_bounds},
and thus close a gap left by~\cite{BernsteinDisserGrossHimburg/20}.
\begin{defn}
	For a directed graph $G=(V,E)$ with source $s\in V$, sinks $T\subseteq V$,
	and arc capacities \mbox{$\mu\colon E\rightarrow\mathbb{\mathbb{R}}_{\geq0}$},
	we define an \textit{$s$-$T$-flow} to be a function $\vartheta\colon E\rightarrow\mathbb{R}_{\geq0}$
	that satisfies
	\begin{alignat*}{3}
		\vartheta(e)\leq\mu(e) & \hspace*{1em} &  & \forall e\in E & \hspace{1em} & \textrm{(capacity constraint)},\\
		\ex_{\vartheta}(v)=0 &  &  & \forall v\in V\setminus(\{s\}\cup T) &  & \textrm{(flow conservation)},\\
		\ex_{\vartheta}(t)\geq0 &  &  & \forall t\in T &  & \textrm{(}T\textrm{ are sinks}\textrm{)},
	\end{alignat*}
	where (using $\delta^{+}(v):=(\{v\}\times V)\cap E$, $\delta^{-}(v):=(V\times\{v\})\cap E$)
	the \textit{excess} of a vertex $v\in V$ is defined as
	\[
	\ex_{\vartheta}(v):=\sum\limits _{e\in\delta^{-}(v)}\vartheta(e)-\sum\limits _{e\in\delta^{+}(v)}\vartheta(e).
	\]
\end{defn}

We extend this notion to multi-commodity flows, where each commodity
has an independent capacity function.
\begin{defn}
	Let $\alpha\in\mathbb{N}$ and $G=(V,E)$ be a graph with $s\in V$
	and $T\subseteq V$. Furthermore, let \mbox{$\vec{\mu}=(\mu_{i}\colon E\to\mathbb{R}_{\geq0})_{i\in\oneto{\alpha}}$}
	be capacity functions. \textit{A}\textit{ multicommodity-flow}
	in $G$ w.r.t.~$\vec{\mu}$ is a tuple $\vec{\vartheta}=(\vartheta_{1},...,\vartheta_{\alpha})$,
	where~$\vartheta_{i}$ is an~$\mbox{\ensuremath{s}-\ensuremath{T}-flow}$
	in $G$ with respect to capacities $\mu_{i}$. The \textit{minimum-excess}
	of the sink vertex $t\in T$ in $\vec{\vartheta}$ is
	\[
	\minex_{\vec{\vartheta}}(t):=\min_{i\in\oneto{\alpha}}\ex_{\vartheta_{i}}(t).
	\]
\end{defn}

For convenience, we let $\mu(u,v):=\mu((u,v))$, $\vartheta(u,v):=\vartheta((u,v))$,
and we let $\mbox{\ensuremath{\ex_{\vartheta}(V'):=\sum_{v\in V'}\ex_{\vartheta}(v)}}$
for $V'\subseteq V$, and $\minex_{\vec{\vartheta}}(T'):=\sum_{t\in T'}\minex_{\vec{\vartheta}}(t)$
for $T'\subseteq T$ in the following.

An instance of the problem $\problem$, for $\alpha\in\mathbb{N}$,
is given by a tuple $(G,s,T,\vec{\mu})$, where $\mbox{\ensuremath{G=(V,E)}}$
is a directed graph, $s\in V$ is a source vertex, $T\subseteq V$
contains sink vertices, and \mbox{$\vec{\mu}=(\mu_{i}\colon E\to\mathbb{R}_{\geq0})_{i\in\oneto{\alpha}}$}
are capacity functions. The problem is to find a subset of sinks $X\subseteq T$
with~$|X|=k$ that maximizes the objective function 
\[
f(X)=\max\limits _{\vec{\vartheta}\in\mathcal{M}_{G,\mu}}\minex_{\vec{\vartheta}}(X),
\]
where $\mathcal{M}_{G,\mu}$ denotes the set of all multicommodity-flows
in $G$ w.r.t. capacities~$\vec{\mu}$.
\begin{example}
	For a prototypical application of \problem, consider a factory
	where $k\in\N$ machines are to be built in a set~$T$ of potential
	locations. Each machine produces the same item and needs a number
	$\alpha\in\N$ of different resources. The output of a machine is
	limited by the resource it has available the least. All resources
	are delivered to the machines along different routes within the factory,
	e.g., some liquids might be transported via pipes, other resources
	might be transported on a conveyor belt or on pallets. The objective
	is to determine in which $k$ locations the machines should be constructed
	in order to maximize overall production.
\end{example}

\begin{thm}
	For every $\alpha\in\mathbb{N}$, the objective of $\problem$ is
	monotone and $\alpha$\textup{-augmentable}.\label{thm:mcflow_augmentable}
\end{thm}

\begin{proof}
	Let $X\subseteq T$ and $t\in T\setminus X$. To prove monotonicity,
	fix some flow $\vec{\vartheta}$ with $\minex_{\vec{\vartheta}}(X)=f(X)$.
	By definition, \mbox{$\minex_{\vec{\vartheta}}(X)\leq\minex_{\vec{\vartheta}}(X\cup\{t\})\leq f(X\cup\{t\})$}
	holds and thus $f$ is monotone.
	
	To show $\alpha$-augmentability, let $(G,s,T,\mu)$ be an instance of
	\textsc{Multi-Sink}\linebreak $\alpha$-\textsc{Commodity Flow}. Let $X,Y\subseteq T$ such that $Y':=Y\setminus X\neq\emptyset$.
	We show that there exists $y\in Y'$ with
	\[
	f(X\cup\{y\})-f(X)\geq\frac{f(X\cup Y')-\alpha f(X)}{\left|Y'\right|}.
	\]
	This suffices because, with
	\[
	\frac{f(X\cup Y')-\alpha f(X)}{\left|Y'\right|}=\frac{f(X\cup Y)-\alpha f(X)}{\left|Y'\right|}\geq\frac{f(X\cup Y)-\alpha f(X)}{\left|Y\right|},
	\]
	$\alpha$-augmentability of the problem follows.
	
	Let $\vec{\vartheta}^{X\cup Y'}=(\vartheta_{1}^{X\cup Y'},...,\vartheta_{\alpha}^{X\cup Y'})$
	be a multicommodity-flow in $G$ that maximises the minimum-exess
	$\minex_{\vec{\vartheta}^{X\cup Y'}}(X\cup Y')$, i.e., $\minex_{\vec{\vartheta}^{X\cup Y'}}(X\cup Y')=f(X\cup Y')$,
	such that $\vartheta_{i}^{X\cup Y'}$ is a maximum $s$-$(X\cup Y')$-flow
	w.r.t. capacity $\mu_{i}$ for all $i\in\oneto{\alpha}$. Such a multicommodity-flow
	can be obtained by augmenting a flow that maximises $\minex_{\vec{\vartheta}^{X\cup Y'}}(X\cup Y')$,
	e.g., with the Edmonds-Karp algorithm (cf.~\cite{KorteVygen2012}).
	Furthermore, we let~$\vec{\vartheta}^{X}=(\vartheta_{1}^{X},...,\vartheta_{\alpha}^{X})$
	be a multicommodity-flow in $G$ with $\ex_{\vartheta_{i}^{X}}(X)=f(X)$
	and $\ex_{\vartheta_{i}^{X}}(T\setminus X)=0$ for all $i\in\oneto{\alpha}$,
	i.e., $\vec{\vartheta}^{X}$ maximises the minimum-excess of the set~$X$ while the values of all flows $\vartheta_{i}^{X}$ are as small
	as possible. This multicommodity-flow can be obtained by reducing
	the flows of a multicommodity-flow that maximises $\minex_{\vec{\vartheta}^{X}}(X\cup Y')$
	along paths of a path decomposition of the flow (cf.~\cite{KorteVygen2012}).
	We define the function~$g\colon X\rightarrow\oneto{\alpha}$, such
	that, for all $x\in X$, no flow $\tilde{\vartheta}$ w.r.t. capacity
	$\mu_{g(x)}$ exists with $\ex_{\tilde{\vartheta}}(x')\geq\ex_{\vartheta_{g(x)}^{X}}(x')$
	for all $x'\in X\setminus\{x\}$ and with $\ex_{\tilde{\vartheta}}(x)>\ex_{\vartheta_{g(x)}^{X}}(x)$.
	This means that the flow $\vartheta_{g(x)}$ is one of the flows limiting
	the value of $\minex_{\vec{\vartheta}^{X}}(x)$. Let $g^{-1}(i)=\{x\in X\mid g(x)=i\}$
	for all $i\in\oneto{\alpha}$ be the preimage of $g$. Obviously
	\begin{equation}
		\bigcup_{i=1}^{\alpha}g^{-1}(i)=X.\label{eq:preimage_whole_X}
	\end{equation}
	
	We add a super sink $t$ to $G$ and let $\tilde{G}=(\tilde{V},\tilde{E})$
	with $\tilde{V}:=V\cup\{t\}$ and \mbox{$\tilde{E}:=E\cup\{(v,t)\mid v\in(X\cup Y')\}$}
	denote the resulting graph. Furthermore, we define the capacity functions
	$\tilde{\mu}_{i}\colon\tilde{E}\rightarrow\mathbb{R}$ for all $i\in\oneto{\alpha}$
	such that, for $(u,v)\in\tilde{E}$,
	\[
	\tilde{\mu}_{i}(u,v):=\begin{cases}
		\mu_{i}(u,v), & \textrm{if }(u,v)\in E,\\
		\max\{\ex_{\vartheta_{i}^{X}}(u),\ex_{\vartheta_{i}^{X\cup Y'}}(u)\}, & \textrm{if }(u,v)\in X\times\{t\},\\
		\ex_{\vartheta_{i}^{X\cup Y'}}(u), & \textrm{if }(u,v)\in Y'\times\{t\}.
	\end{cases}
	\]
	Now we extend the flow $\vec{\vartheta}^{X\cup Y'}$ to a flow $\tilde{\vec{\vartheta}}^{X\cup Y'}$
	in $\tilde{G}$, such that, for all $i\in\oneto{\alpha}$ and $(u,v)\in\tilde{E}$,
	\[
	\tilde{\vartheta}_{i}^{X\cup Y'}(u,v):=\begin{cases}
		\vartheta_{i}^{X\cup Y'}(u,v), & \textrm{if }(u,v)\in E,\\
		\ex_{\vartheta_{i}^{X\cup Y'}}(u), & \mathrm{else},
	\end{cases}
	\]
	holds, and analogously, we extend the flow $\vec{\vartheta}^{X}$
	to a flow $\tilde{\vec{\vartheta}}^{X}$ in $\tilde{G}$. With this
	definition, $\tilde{\vartheta}_{i}^{X\cup Y'}$ is a maximum $s$-$t$-flow
	w.r.t. capacity $\tilde{\mu}_{i}$, because $\vartheta_{i}^{X\cup Y'}$
	is a maximum~\mbox{$s$-$(X\cup Y')$-flow} w.r.t. capacity $\mu_{i}$.
	
	For $i\in\oneto{\alpha}$, let $\tilde{\vartheta}_{i}$ be a maximum
	$s$-$t$-flow w.r.t. capacity $\tilde{\mu}_{i}$ in $\tilde{G}$
	obtained from $\tilde{\vartheta}_{i}^{X}$ by using the Edmonds-Karp
	algorithm. Then its value is as large as the value of the flow $\tilde{\vartheta}_{i}^{X\cup Y'}$.
	We project $\tilde{\vartheta}_{i}$ onto a flow in $G$, i.e., we
	set $\vartheta_{i}:=\tilde{\vartheta}_{i}|_{E}$ for $i\in\oneto{\alpha}$
	and define $\vec{\vartheta}:=(\vartheta_{1},...,\vartheta_{\alpha})$.
	By definition of $g$ and $\vec{\vartheta}^{X}$, we have
	\begin{equation}
		\ex_{\vartheta_{g(x)}}(x)=\ex_{\vartheta_{g(x)}^{X}}(x)\label{eq:restriction_from_g}
	\end{equation}
	for all $x\in X$. Because $\tilde{\vartheta}_{i}$ is a maximum $s$-$t$-flow
	in $\tilde{G}$ w.r.t. capacity $\tilde{\mu}_{i}$, $\vartheta_{i}$
	is a maximum $s$-$(X\cup Y')$-flow w.r.t. capacity $\mu_{i}$ in
	$G$. Since $\vartheta_{i}^{X\cup Y'}$ is also a maximum $s$-$(X\cup Y')$-flow
	w.r.t. capacity $\mu_{i}$, we have 
	\begin{equation}
		\ex_{\vartheta_{i}}(X\cup Y')=\ex_{\vartheta_{i}^{X\cup Y'}}(X\cup Y').\label{eq:f_max_flow}
	\end{equation}
	For all $x\in X$, we know that the excess of $x$ in $\vartheta_{i}$
	is as large as the flow $\tilde{\vartheta}_{i}(x,t)$, i.e., 
	\begin{equation}
		\ex_{\vartheta_{i}}(x)=\tilde{\vartheta}_{i}(x,t)\leq\tilde{\mu}_{i}(x,t)=\max\{\ex_{\vartheta_{i}^{X}}(x),\ex_{\vartheta_{i}^{X\cup Y'}}(x)\}\leq\ex_{\vartheta_{i}^{X}}(x)+\ex_{\vartheta_{i}^{X\cup Y'}}(x).\label{eq:excess_leq_edge_to_supersink}
	\end{equation}
	By maximality of $\vec{\vartheta}^{X}$ and definition of $\vec{\vartheta}$,
	\begin{equation}
		\minex_{\vec{\vartheta}}(X)=\minex_{\vec{\vartheta}^{X}}(X)=f(X)\label{eq:new_flow_full}
	\end{equation}
	holds. Since $X\cap Y'=\emptyset$, we obtain
	\begin{eqnarray}
		& &\ex_{\vartheta_{i}^{X\cup Y'}}(Y')-\textrm{ex}_{\vartheta_{i}}(Y')\nonumber \\
		& = & \ex_{\vartheta_{i}^{X\cup Y'}}(X\cup Y')-\ex_{\vartheta_{i}}(X\cup Y')-\ex_{\vartheta_{i}^{X\cup Y'}}(X)+\ex_{\vartheta_{i}}(X)\nonumber \\
		& \stackrel{\eqref{eq:f_max_flow}}{=} & \ex_{\vartheta_{i}}(X)-\ex_{\vartheta_{i}^{X\cup Y'}}(X)\nonumber \\
		& = & \sum\limits _{x\in X\setminus g^{-1}(i)}(\ex_{\vartheta_{i}}(x)-\ex_{\vartheta_{i}^{X\cup Y'}}(x))+\sum\limits _{x\in g^{-1}(i)}(\ex_{\vartheta_{i}}(x)-\ex_{\vartheta_{i}^{X\cup Y'}}(x))\nonumber \\
		& \stackrel{\eqref{eq:excess_leq_edge_to_supersink},\eqref{eq:restriction_from_g}}{\leq} & \sum\limits _{x\in X\setminus g^{-1}(i)}\ex_{\vartheta_{i}^{X}}(x)+\sum\limits _{x\in g^{-1}(i)}(\ex_{\vartheta_{i}^{X}}(x)-\ex_{\vartheta_{i}^{X\cup Y'}}(x))\nonumber \\
		& = & f(X)-\sum\limits _{x\in g^{-1}(i)}\ex_{\vartheta_{i}^{X\cup Y'}}(x),\label{eq:first_ineq}
	\end{eqnarray}
	where we used minimality of $\vec{\vartheta}^{X}$. Using this we
	can compute
	\begin{eqnarray}
		\minex_{\vec{\vartheta}^{^{X\cup Y'}}}(Y') & = & \sum\limits _{y\in Y'}\min_{i\in\oneto{\alpha}}\{\ex_{\vartheta_{i}^{X\cup Y'}}(y)\}\nonumber \\
		& = & \sum\limits _{y\in Y'}\min_{i\in\oneto{\alpha}}\{\ex_{\vartheta_{i}}(y)+(\ex_{\vartheta_{i}^{X\cup Y'}}(y)-\ex_{\vartheta_{i}}(y))\}\nonumber \\
		& \leq & \sum\limits _{y\in Y'}\Bigl(\min_{i\in\oneto{\alpha}}\{\ex_{\vartheta_{i}}(y)\}+\sum\limits _{i=1}^{\alpha}(\ex_{\vartheta_{i}^{X\cup Y'}}(y)-\ex_{\vartheta_{i}}(y))\Bigr)\nonumber \\
		& = & \minex_{\vec{\vartheta}}(Y')+\sum\limits _{i=1}^{\alpha}(\ex_{\vartheta_{i}^{X\cup Y'}}(Y')-\ex_{\vartheta_{i}}(Y'))\nonumber \\
		& \stackrel{\eqref{eq:first_ineq}}{\leq} & \minex_{\vec{\vartheta}}(Y')+\sum\limits _{i=1}^{\alpha}\bigl(f(X)-\sum\limits _{x\in g^{-1}(i)}\ex_{\vartheta_{i}^{X\cup Y'}}(x)\bigr)\nonumber \\
		& \stackrel{\eqref{eq:preimage_whole_X}}{=} & \minex_{\vec{\vartheta}}(Y')+\alpha f(X)-\sum_{x\in X}\ex_{\vartheta_{g(x)}^{X\cup Y'}}(x).\label{eq:second_ineq}
	\end{eqnarray}
	Finally, because of $X\cap Y'=\emptyset$, we get
	\begin{eqnarray*}
		f(X\cup Y') & = & \minex_{\vec{\vartheta}^{^{X\cup Y'}}}(X)+\minex_{\vec{\vartheta}^{^{X\cup Y'}}}(Y')\\
		& = & \sum\limits _{x\in X}\bigl(\min_{i\in\oneto{\alpha}}\ex_{\vartheta_{i}^{X\cup Y'}}(x)\bigr)+\minex_{\vec{\vartheta}^{X\cup Y'}}(Y')\\
		& \stackrel{\eqref{eq:second_ineq}}{\leq} & \sum\limits _{x\in X}\ex_{\vartheta_{g(x)}^{X\cup Y'}}(x)+\minex_{\vec{\vartheta}}(Y')+\alpha f(X)-\sum_{x\in X}\ex_{\vartheta_{g(x)}^{X\cup Y'}}(x)\\
		& = & \sum\limits _{y\in Y'}\minex_{\vec{\vartheta}}(y)+\alpha f(X),
	\end{eqnarray*}
	which is equivalent to
	\begin{equation}
		\sum\limits _{y\in Y'}\minex_{\vec{\vartheta}}(y)\geq f(X\cup Y')-\alpha f(X).\label{eq:almost_done}
	\end{equation}
	
	Now, we show that $f(X\cup\{y\})-f(X)\geq\minex_{\vec{\vartheta}}(y)$
	for all $y\in Y'$, which completes the proof, because then
	\begin{eqnarray*}
	\left|Y'\right|\bigl(\max\limits _{y\in Y'}f(X\cup\{y\})-f(X)\bigr)&\geq&\sum_{y\in Y'}\bigl(f(X\cup\{y\})-f(X)\bigr) \\
	&\geq&\sum_{y\in Y'}\minex_{\vec{\vartheta}}(y)\stackrel{\eqref{eq:almost_done}}{\geq}f(X\cup Y')-\alpha f(X).
	\end{eqnarray*}
	To show this, take any $y\in Y'$. Since $X\cap Y'=\emptyset$, we
	know that
	\[
	\minex_{\vec{\vartheta}}(X\cup\{y\})=\minex_{\vec{\vartheta}}(X)+\minex_{\vec{\vartheta}}(y)\stackrel{\eqref{eq:new_flow_full}}{=}f(X)+\minex_{\vec{\vartheta}}(y).
	\]
	Furthermore, we have $f(X\cup\{y\})\geq\minex_{\vec{\vartheta}}(X\cup\{y\})$ because $\vec{\vartheta}$ is a multicommodity-flow in $G$.
	Combining these two insights yields \mbox{$f(X\cup\{y\})-f(X)\geq\minex_{\vec{\vartheta}}(y)$}.
\end{proof}

For $\alpha=2$, $\problem$ problem is equivalent to the $\textsc{BridgeFlow}$
problem considered in~\cite{BernsteinDisserGrossHimburg/20}. We
generalize the tight lower bound construction for $\textsc{BridgeFlow}$
to arbitrary $\alpha\in\mathbb{N}$.

To show a lower bound on the approximation ratio of the greedy algorithm,
we construct a family of instances of the $\problem$ problem. For
$k\in\mathbb{N}$, $k\geq2$ we let $x:=\frac{k}{k-1}$. Now, we define
the graphs $G_{k}=(V_{k},E_{k})$ (cf.~Figure~\ref{fig:lowerbound_graph})
via
\begin{eqnarray*}
	V_{k} & := & \{s,v_{1},...,v_{\alpha k},t_{1},...,t_{2\alpha k}\},\\
	E_{k} & := & \bigcup\limits _{i=1}^{\alpha}\mathcal{E}_{k,i},\\
	\mathcal{E}_{k,i} & := & E_{k,i}^{1}\cup E_{k,i}^{\infty}\cup\bigcup_{j=1}^{\alpha k}E_{k,i,j}\cup\bigcup_{j=1}^{\alpha k}E'_{k,i,j},\\
	E_{k,i}^{1} & := & \{(s,t_{(\alpha+i-1)k+1}),...,(s,t_{(\alpha+i)k})\},\\
	E_{k,i}^{\infty} & := & \{(s,t_{\alpha k+1}),...,(s,t_{2\alpha k})\}\setminus E_{k,i}^{1},\\
	E_{k,i,j} & := & \{(s,v_{j}),(v_{j},t_{j})\}\forall j\in\oneto{\alpha k},\\
	E'_{k,i,j} & := & \{(v_{j},t_{(\alpha+i-1)k+1}),...,(v_{j},t_{(\alpha+i)k})\}\forall j\in\oneto{\alpha k},
\end{eqnarray*}
capacity functions $\mu^{k}=(\mu_{1}^{k},...,\mu_{\alpha}^{k})$ with
$\mu_{i}^{k}\colon E^{k}\rightarrow\R$ for $i\in\oneto{\alpha}$
and
\[
\mu_{i}^{k}(e)=\begin{cases}
	1, & \textrm{if }e\in E_{k,i}^{1},\\
	\infty, & \textrm{if }e\in E_{k,i}^{\infty},\\
	x^{\alpha k-j+1}, & \textrm{if }e\in E_{k,i,j}\textrm{ for some }j\in\oneto{\alpha k},\\
	\frac{1}{k}x^{\alpha k-j+1}, & \textrm{if }e\in E'_{k,i,j}\textrm{ for some }j\in\oneto{\alpha k},\\
	0, & \textrm{else.}
\end{cases}
\]
Note that only the arcs in $\mathcal{E}_{k,i}$ allow a flow of commodity
$i$. We define $s$ to be the source vertex and $T:=\{t_{1},...,t_{2\alpha k}\}$
to be the set of sink vertices.

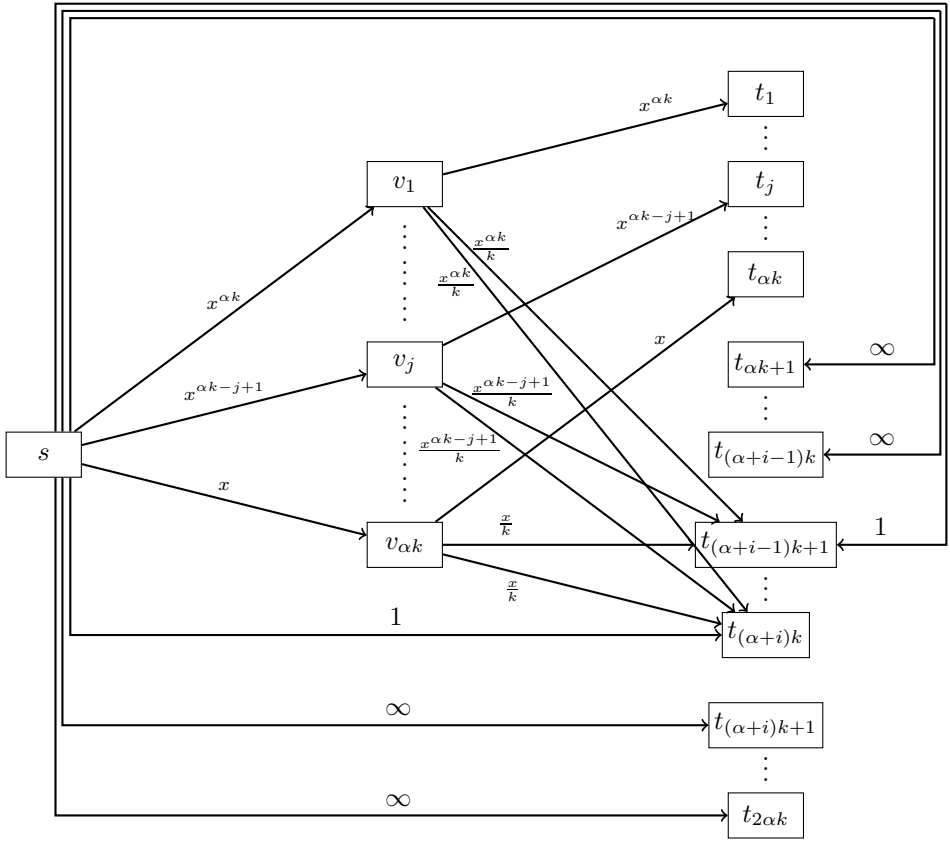
\begin{figure}
	\begin{center}
		\begin{tikzpicture}[yscale=1.2, xscale=.8]
			\tikzstyle{node}=[draw,rectangle, inner sep=2pt,minimum width=10mm,minimum height=6mm] 
			\tikzstyle{edge}=[thick,->]
			\node[node] (s1) at (0,4) {$s$};
			\node[node] (v1) at (6,7) {$v_1$};
			\node[node] (vj) at (6,5) {$v_j$};
			\node[node] (v2) at (6,3) {$v_{\alpha k}$};
			\node[node] (t1) at (12,8) {$t_1$};
			\node[node] (tj) at (12,7) {$t_j$};
			\node[node] (t2) at (12,6) {$t_{\alpha k}$};
			\node[node] (t3) at (12,5) {$t_{\alpha k+1}$};
			\node[node] (t4) at (12,4) {$t_{(\alpha +i-1)k}$};
			\node[node] (t5) at (12,3) {$t_{(\alpha +i-1)k+1}$};
			\node[node] (t6) at (12,2) {$t_{(\alpha +i)k}$};
			\node[node] (t7) at (12,1) {$t_{(\alpha +i)k+1}$};
			\node[node] (t8) at (12,0) {$t_{2\alpha k}$};
			\node (dots) at (6,6.5) {$\vdots$};
			\node (dots) at (6,6.1) {$\vdots$};
			\node (dots) at (6,5.7) {$\vdots$};
			\node (dots) at (6,4.5) {$\vdots$};
			\node (dots) at (6,4.1) {$\vdots$};
			\node (dots) at (6,3.7) {$\vdots$};
			\node (dots) at (12,7.6) {$\vdots$};
			\node (dots) at (12,6.6) {$\vdots$};
			\node (dots) at (12,4.6) {$\vdots$};
			\node (dots) at (12,2.6) {$\vdots$};
			\node (dots) at (12,0.6) {$\vdots$};
			
			\draw[edge] (s1) -- node[above] {$\scriptstyle x^{\alpha k}$} (v1);
			\draw[edge] (s1) -- node[above] {$\scriptstyle x^{\alpha k-j+1}$} (vj);
			\draw[edge] (s1) -- node[above] {$\scriptstyle x$} (v2);
			\draw[edge] (v1) -- node[above,pos=0.75] {$\scriptstyle x^{\alpha k}$} (t1);
			\draw[edge] (v1) -- node[above,pos=0.2] {$\scriptstyle \frac{x^{\alpha k}}{k}$} (t5);
			\draw[edge] (v1) -- node[below=4pt,pos=0.1] {$\scriptstyle \frac{x^{\alpha k}}{k}$} (t6);
			\draw[edge] (vj) -- node[above,pos=0.75] {$\scriptstyle x^{\alpha k-j+1}$} (tj);
			\draw[edge] (vj) -- node[above=1pt,pos=0.25] {$\scriptstyle \frac{x^{\alpha k-j+1}}{k}$} (t5);
			\draw[edge] (vj) -- node[below=7pt,pos=0.08] {$\scriptstyle \frac{x^{\alpha k-j+1}}{k}$} (t6);
			\draw[edge] (v2) -- node[above,pos=0.75] {$\scriptstyle x$} (t2);
			\draw[edge] (v2) -- node[above,pos=0.25] {$\scriptstyle \frac{x}{k}$} (t5);
			\draw[edge] (v2) -- node[below,pos=0.25] {$\scriptstyle \frac{x}{k}$} (t6);
			
			\draw[thick] (s1.63) |- (15,9);
			\draw[thick] (s1.51) |- (14.9,8.92);
			\draw[thick] (s1.41) |- (14.8,8.84);
			\draw[edge] (15,9) |- node[above,pos=0.8] {$1$} (t5);
			\draw[edge] (14.9,8.92) |- node[above,pos=0.75] {$\infty$} (t4);
			\draw[edge] (14.8,8.84) |- node[above,pos=0.7] {$\infty$} (t3);
			\draw[edge] (s1.319) |- node[above,pos=0.75] {$1$} (t6);
			\draw[edge] (s1.309) |- node[above,pos=0.76] {$\infty$} (t7);
			\draw[edge] (s1.297) |- node[above,pos=0.755] {$\infty$} (t8);
		\end{tikzpicture}
	\end{center}
	
	\caption{\label{fig:lowerbound_graph}The graph $G_{k}$ only with edges from
		$\mathcal{E}_{k,i}$ and capacities $\mu_{i}^{k}$.}
\end{figure}

In the following proof we will need the following observation: Using
$x=\frac{k}{k-1}$ and with $n\in\mathbb{N}$ the equation
\begin{eqnarray}
	1+\frac{1}{k}\sum\limits _{j=1}^{n}x^{j} & = & 1+\frac{1}{k}\Bigl(\frac{x^{n+1}-1}{x-1}-1\Bigr)=1+\frac{1}{k}\Bigl(\frac{\bigl(\frac{k}{k-1}\bigr)^{n+1}-1}{\bigl(\frac{k}{k-1}\bigr)-1}-1\Bigr)\nonumber \\
	& = & 1+\frac{1}{k}\Bigl((k-1)\Bigl(\Bigl(\frac{k}{k-1}\Bigr)^{n+1}-1\Bigr)-1\Bigr)\nonumber \\
	& = & 1+\Bigl(\frac{k}{k-1}\Bigr)^{n}-\frac{1}{k}((k-1)+1) = \Bigl(\frac{k}{k-1}\Bigr)^{n}=x^{n}\label{eq:sum_capacity_good_vertices}
\end{eqnarray}
holds.

We will now show in which order the greedy algorithm picks the vertices
from the set $T$. We assume that the tie-breaking works out in our
favor. This can be achieved by introducing small offsets to the capacities.
For better readability we omit this here.
\begin{lem}
	\label{lem:greedy_pick_order}Let $\alpha,k\in\N$. In iteration $\ell\in\oneto{\alpha k}$,
	the greedy algorithm picks sink vertex $t_{\ell}$. Furthermore, a
	multicommodity-flow $\vec{\vartheta}=(\vartheta_{1},...,\vartheta_{\alpha})$
	with maximum minimum-excess of the vertices $\{t_{1},...,t_{\ell}\}$
	for $\ell\in\oneto{\alpha k}$ always satisfies \mbox{$\vartheta_{i}(e)=x^{\alpha k-j+1}$}
	for all $e\in E_{k,i,j}$ with $i\in\oneto{\alpha}$ and $j\in\oneto{\ell}$,
	i.e., all arcs in $E_{k,i,j}$ are fully saturated for all $i\in\oneto{\alpha}$
	and $j\in\oneto{\ell}$.
\end{lem}

\begin{proof}
	We will prove the statement by induction. In iteration $\ell=1$,
	the gain of picking vertex $t_{j}$ with $j\in\oneto{\alpha k}$ is
	$x^{\alpha k-j+1}$, because for all $i\in\oneto{\alpha}$ we can
	have a flow of value $x^{\alpha k-j+1}$ of commodity $i$ from $s$
	via $v_{j}$ and the edges in $E_{k,i,j}$ to $t_{j}$ and no more
	flows to $t_{j}$ are possible, since the only incoming arc to $t_{j}$,
	which allows a flow of commodity $i$, is the arc $(v_{j},t_{j})\in E_{k,i,j}$.
	The gain of picking vertex $t_{j}$ with $j\in\{\alpha k+1,...,2\alpha k\}$
	is the minimum of all commodities flowing to $t_{j}$ and there is
	only one commodity which does not allow an unbounded flow to $t_{j}$,
	because for $i\in\oneto{\alpha}\setminus\{\left\lceil \frac{j-\alpha}{k}\right\rceil \}$
	there is an arc from $s_{i}$ to $t_{j}$ in $E'_{k,i,j}$ with infinite
	capacity for commodity $i$. The maximum flow of the commodity with
	a finite flow to $t_{j}$ is
	\begin{align*}
		1+\frac{1}{k}\sum\limits _{j=1}^{\alpha k}x^{j} & \stackrel{\eqref{eq:sum_capacity_good_vertices}}{=}x^{\alpha k},
	\end{align*}
	and, thus, with proper tie-breaking, the greedy algorithm chooses
	vertex $t_{1}$. For $i\in\oneto{\alpha}$, the only incoming path
	that allows a flow of commodity $i$ from $s$ to $t_{1}$ is along
	the edges in $E_{k,i,1}$, so they have to be fully saturated by a
	multicommodity-flow with maximum minimum-excess.
	
	Now suppose the statement is true for some $\ell\in\oneto{\alpha k-1}$,
	i.e., the greedy algorithm has picked edges~$t_{1},...,t_{\ell}$
	and a multicommodity-flow with maximum minimum-excess of the vertices
	$\{t_{1},...,t_{\ell}\}$ fully saturates all arcs in $E_{k,i,j}$
	for all $i\in\oneto{\alpha}$ and $j\in\oneto{\ell}$. Then the gain
	of picking vertex $t_{j}$ for $j\in\{\ell+1,...,\alpha k\}$ is still~$x^{\alpha k-j+1}$,
	because all $s$-$t_{j}$-paths for $i\in\oneto{\alpha}$ do not carry
	flow that contributes to the maximum minimum-excess. The gain of picking
	vertex $t_{j}$ for $j\in\{\alpha k+1,...,2\alpha k\}$ is still the
	minimum of all commodities flowing to $t_{j}$, and again there is
	only one commodity which does not allow an unbounded flow to $t_{j}$.
	Because all incoming flow at vertices $v_{1},...,v_{\ell}$ already
	saturates all incoming arcs, there is no flow of this commodity via
	a vertex in~$\{v_{1},...,v_{\ell}\}$ to $t_{j}$ possible without
	reducing the minimum-excess of another sink vertex by the same amount.
	Thus, the maximal flow of this commodity to $t_{j}$ is
	\begin{align*}
		1+\frac{1}{k}\sum\limits _{j=1}^{\alpha k-\ell}x^{j} & \stackrel{\eqref{eq:sum_capacity_good_vertices}}{=}x^{\alpha k-\ell},
	\end{align*}
	so, with proper tie-breaking, the greedy algorithm picks vertex $t_{\ell+1}$
	next. For $i\in\oneto{\alpha}$, the only incoming path that allows
	a flow of commodity $i$ from $s$ to $t_{j}$ for $j\in\oneto{\ell}$
	is along the edges in $E_{k,i,j}$, so they have to be fully saturated
	by a multicommodity-flow with maximum minimum-excess.
\end{proof}

With this, we obtain a lower bound for the approximation ratio of
the greedy algorithm on $\augfuncs_{\alpha}$ for~$\alpha\in\mathbb{N}$
that tightly matches the upper bound of~\cite{BernsteinDisserGrossHimburg/20},
i.e., we obtain Corollary~\ref{cor:augmentable_lower_bound} for
$\alpha\in\mathbb{N}$. In particular, it follows that the objective
of $\problem$ is not $\beta$-augmentable for any $\beta<\alpha$.
We will generalize the lower bound to all $\alpha\geq1$ in Section~\ref{subsec:lower_bounds}.
\begin{thm}
	\label{thm:lower_bound_mcflows}For $\alpha\in\mathbb{N}$, the greedy
	algorithm has an approximation ratio of at least $\alpha\frac{\e^{\alpha}}{\e^{\alpha}-1}$
	for $\problem$.
\end{thm}

\begin{proof}
	By Lemma \ref{lem:greedy_pick_order}, the greedy algorithm picks
	the sinks $t_{1},...,t_{\alpha k}$ in the first~$\alpha k$ iterations
	and the objective increases by $x^{\alpha k-j+1}$ when sink vertex
	$t_{j}$ is picked and thus the minimum-excess of the greedy solution
	is
	\[
	f(\greedyset_{k})=\sum_{j=1}^{\alpha k}x^{j}\stackrel{\eqref{eq:sum_capacity_good_vertices}}{=}k(x^{\alpha k}-1).
	\]
	We compare this to the solution that picks the vertices $t_{\alpha k+1},...,t_{2\alpha k}$
	(which is, in fact, an optimal solution for cardinality $\alpha k$).
	Increasing the flow to one of these vertices does not reduce the flow
	to the others, so the minimum-excess of any of these vertices is
	\begin{align*}
		1+\frac{1}{k}\sum\limits _{j=1}^{\alpha k}x^{j} & \stackrel{\eqref{eq:sum_capacity_good_vertices}}{=}x^{\alpha k},
	\end{align*}
	and their total minimum-excess thus is $\alpha kx^{\alpha k}$. Using
	this and $x=\frac{k}{k-1}$, we calculate the ratio between this solution
	and the greedy solution and get
	\[
	\frac{\alpha kx^{\alpha k}}{k(x^{\alpha k}-1)}=\frac{\alpha x^{\alpha k}}{x^{\alpha k}-1}=\frac{\alpha\bigl(\frac{k}{k-1}\bigr)^{\alpha k}}{\bigl(\frac{k}{k-1}\bigr)^{\alpha k}-1}=\frac{\alpha\bigl(\bigl(\frac{k}{k-1}\bigr)^{k}\bigr)^{\alpha}}{\bigl(\bigl(\frac{k}{k-1}\bigr)^{k}\bigr)^{\alpha}-1}.
	\]
	Using the identity $\lim_{k\rightarrow\infty}(k/(k-1))^{k}=\e$, we
	obtain the limit\setlength{\belowdisplayskip}{-12pt}
	\[
	\lim\limits _{k\rightarrow\infty}\frac{\alpha\bigl(\bigl(\frac{k}{k-1}\bigr)^{k}\bigr)^{\alpha}}{\bigl(\bigl(\frac{k}{k-1}\bigr)^{k}\bigr)^{\alpha}-1}=\alpha\frac{\e^{\alpha}}{\e^{\alpha}-1}.
	\]
\end{proof}

\subsection{Separating Function Classes}

We are now ready to show Proposition~\ref{prop:separation} for $\alpha\in\mathbb{N}\setminus\{1\}$.
The case $\alpha\geq1$ will be addressed in Section~\ref{subsec:lower_bounds}.

We first separate $\augfuncs_{\alpha}$ for $\alpha\in\mathbb{N}\setminus\{1\}$
by showing that the objective of $\problem$ does not have a (weak)
submodularity ratio bounded away from zero, and cannot be represented
as the weighted rank function of some independence system.
\begin{prop}
	For every $\gamma,q\in(0,1)$, and $\alpha\in\mathbb{N}$ with $\alpha\geq2$,
	it holds that $\mbox{\ensuremath{\augfuncs_{\alpha}\nsubseteq(\srfuncs_{\gamma}\cup\rqfuncs_{q})}}$.\label{prop:separating_alpha_aug_lite}
\end{prop}

\begin{proof}
	In order to prove that $\augfuncs_{\alpha}\nsubseteq\srfuncs_{\gamma}$,
	by Theorem~\ref{thm:mcflow_augmentable}, it suffices to construct
	an instance of $\problem$ that has weak submodularity ratio smaller
	$\gamma$. Let
	\begin{eqnarray*}
		T & := & \{t_{1},t_{2},t_{3}\},\\
		V & := & \{s,v_{1},v_{2}\}\cup T,\\
		E & := & \{(s,v_{1}),(s,v_{2}),(s,t_{1}),(s,t_{3}),(v_{1},t_{1}),(v_{1},t_{2}),(v_{2},t_{2}),(v_{2},t_{3})\},\\
		G & := & (V,E)
	\end{eqnarray*}
	and 
	\[
	\mu\colon E\rightarrow\R_{\geq0}^{\alpha},\mu(e)=\begin{cases}
		(1,0,0,...,0), & \textrm{if }e\in\{(s,v_{1}),(s,t_{3}),(v_{1},t_{1}),(v_{1},t_{2})\},\\
		(0,1,1,...,1), & \textrm{else}.
	\end{cases}
	\]
	A diagram of the graph can be seen in Figure \ref{fig:alpha-commodity-flow_submod-ratio_0}.
	With proper tie breaking (or by adding small extra capacities), the
	greedy algorithm picks the sink $t_{2}$ in the first iteration. Adding
	any other sink to this solution does not increase the objective value,
	i.e., for all $t\in T$, we have $\mbox{\ensuremath{\sum_{t\in T}(f(\greedyset_{1}\cup\{t\})-f(\greedyset_{1}))=0}}$.
	But since $\mbox{\ensuremath{f(\greedyset_{1}\cup\{t_{1},t_{3}\})-f(\greedyset_{1})=1}}$,
	the weak submodularity ratio of this problem is $0<\gamma$.
	
	To prove that $\augfuncs_{\alpha}\nsubseteq\rqfuncs_{q}$, we define
	an instance of $\problem$ by
	\begin{eqnarray*}
		T & := & \{t_{1},t_{2}\},\\
		V & := & \{s,v\}\cup T,\\
		E & := & \{(s,v),(v,t_{1}),(v,t_{2})\},\\
		G & := & (V,E)
	\end{eqnarray*}
	and 
	\[
	\mu\colon E\rightarrow\R_{\geq0}^{\alpha},\mu(e)=\begin{cases}
		(3,...,3), & \textrm{if }e=(s,v),\\
		(2,...,2), & \textrm{else}.
	\end{cases}
	\]
	We have $f(\emptyset)=0$, $f(\{t_{1}\})=f(\{t_{2}\})=2$ and $f(\U)=3$.
	If $f$ could be modelled as the weighted rank function of some independence
	system, the corresponding weight function would have to satisfy $w(t_{1})=w(t_{2})=2$
	and simultaneously $w(t_{1})+w(t_{2})=3$ or $\max\{w(t_{1}),w(t_{2})\}=3$.
	Since this is impossible, $f$ cannot be modelled as the weighted
	rank function of some independence system. By Theorem~\ref{thm:mcflow_augmentable},
	this implies $f\in\augfuncs_{\alpha}\setminus\rqfuncs_{q}$.
\end{proof}
\begin{figure}
	\begin{center}
		\begin{tikzpicture}
			\tikzstyle{node}=[draw,rectangle, inner sep=2pt,minimum width=10mm,minimum height=6mm] 
			\tikzstyle{edge}=[thick,->]
			\node[node] (s1) at (0,4) {$s$};
			\node[node] (v1) at (3,5) {$v_1$};
			\node[node] (v2) at (3,3) {$v_2$};
			\node[node] (t1) at (6,6) {$t_1$};
			\node[node] (t2) at (6,4) {$t_2$};
			\node[node] (t3) at (6,2) {$t_3$};
			
			\draw[edge]  (s1) to node[above] {$(1,0)$} (v1);
			\draw[edge] (s1) to node[above] {$(0,1)$} (v2);
			\draw[edge,bend left] (s1) to node[above] {$(0,1)$} (t1);
			\draw[edge,bend right] (s1) to node[above] {$(1,0)$} (t3);
			\draw[edge] (v1) to node[above] {$(1,0)$} (t1);
			\draw[edge] (v1) to node[above] {$(1,0)$} (t2);
			\draw[edge] (v2) to node[above] {$(0,1)$} (t2);
			\draw[edge] (v2) to node[above] {$(0,1)$} (t3);
		\end{tikzpicture}
	\end{center}
	
	\caption{\label{fig:alpha-commodity-flow_submod-ratio_0}An instance of $\protect\problem$
		for $\alpha=2$ with (weak) submodularity ratio 0.}
\end{figure}

We proceed to show the second and third part of Proposition~\ref{prop:separation}
(for all $\alpha\geq1$).
\begin{prop}
	For every $\gamma,q\in(0,1)$, $\alpha\geq1$, it holds that $\srfuncs_{\gamma}\nsubseteq(\augfuncs_{\alpha}\cup\rqfuncs_{q})$.\label{prop:separating_submod_ratio}
\end{prop}

\begin{proof}
	Consider the set $\U=\{a,b\}$ and the objective function
	\[
	f^{\gamma}\colon2^{\U}\rightarrow\R_{\geq0},f^{\gamma}(X)=\begin{cases}
		|X|, & \textrm{if }|X|\leq1,\\
		\frac{2}{\gamma}, & \textrm{else}.
	\end{cases}
	\]
	
	If $f^{\gamma}$ could be modelled as the weighted rank function of
	an independence system $(\U,\I)$, then we would have $\U\in\I$ because
	$f(\U)>f(X)$ for all $X\subsetneq\U$. Then $\I=2^{\U}$ and $f^{\gamma}$
	would be linear which is not true. Thus $f^{\gamma}$ cannot be modelled
	as the weighted rank function of an independence system, and $f^{\gamma}\notin\rqfuncs_{q}$.
	
	Furthermore, $f^{\gamma}\notin\augfuncs_{\alpha}$. To see this, consider
	$X=\emptyset$ and $Y=\{a,b\}$. Then we have $f^{\gamma}(X\cup\{y\})-f^{\gamma}(X)=1$
	for all $y\in Y$, and we have $\frac{f^{\gamma}(X\cup Y)-\alpha f^{\gamma}(X)}{|Y|}=\frac{1}{\gamma}$.
	Since $\gamma<1$, the problem is not $\alpha$-augmentable.
	
	Now, let $X,Y\subseteq\U$ with $X\cap Y=\emptyset$. For $Y=\emptyset$
	the ratio in the definition of the weak submodularity ratio is $\frac{0}{0}=1$.
	Thus, assume $|Y|\geq1$. If $X=\emptyset$, we have
	\[
	\frac{\sum_{y\in Y}f^{\gamma}(X\cup\{y\})-f^{\gamma}(X)}{f^{\gamma}(X\cup Y)-f^{\gamma}(X)}=\frac{|Y|}{f^{\gamma}(Y)}\in\{1,\gamma\}.
	\]
	Otherwise, if $|X|=1$, then $|Y|=1$ and the ratio in the definition
	of the (weak) submodularity ratio is 1.
	In both cases, the ratio is at least $\gamma$, thus the (weak) submodularity
	ratio of this problem is $\gamma$, and $f^{\gamma}\in\srfuncs_{\gamma}$.
\end{proof}

\begin{prop}
	For every $\gamma,q\in(0,1)$, $\alpha\geq1$, it holds that $\rqfuncs_{q}\nsubseteq(\srfuncs_{\gamma}\cup\augfuncs_{\alpha})$.\label{prop:separating_rank_quot}
\end{prop}

\begin{proof}
	We fix $m,n\in\mathbb{N}$ with $q\leq\frac{m}{n}<1$ and $\alpha\geq1$.
	Let
	\begin{eqnarray*}
		A & := & \{a_{1},...,a_{\lceil\alpha\rceil n}\},\\
		B & := & \{b_{1},...,b_{\lceil\alpha\rceil n}\},\\
		C & := & \{c\}\\
		\U & := & A\cup B\cup C,\\
		\I & := & 2^{A}\cup2^{B}\cup\{X\subset\U\mid|X|\leq\lceil\alpha\rceil m\}.
	\end{eqnarray*}
	We consider the independence system $(\U,\I)$ and the weight function
	$w\colon\U\rightarrow\mathbb{R}_{\geq0}$ defined by
	\[
	w(e)=\begin{cases}
		1, & e\in A,\\
		\lceil\alpha\rceil(n-m)+1, & \textrm{else}.
	\end{cases}
	\]
	The weighted rank function $f$ is given by
	\[
	f^{q}\colon\U\rightarrow\R_{\geq0},f^{q}(X)=\max\{w(Y)\mid Y\in2^{X}\cap\I\}.
	\]
	Obviously we have $q(\U,\I)=\frac{m}{n}$, i.e., $f^{q}\in\rqfuncs_{q}$.
	
	For $X=A$, $Y=B$ and $y\in Y$, we calculate
	\begin{eqnarray*}
		f^{q}(X) & = & \lceil\alpha\rceil n,\\
		f^{q}(X\cup\{y\}) & = & \max\{\lceil\alpha\rceil n,\lceil\alpha\rceil(n-m)+1+(\lceil\alpha\rceil m-1)\}=\lceil\alpha\rceil n,\\
		f^{q}(X\cup Y) & = & \lceil\alpha\rceil n(\lceil\alpha\rceil(n-m)+1).
	\end{eqnarray*}
	Suppose, $f^{q}$ was $\alpha$-augmentable. Then
	\[
	f^{q}(X\cup\{y\})-f^{q}(X)\geq\frac{f^{q}(X\cup Y)-\alpha f^{q}(X)}{|Y|},
	\]
	i.e.,
	\[
	\lceil\alpha\rceil n-\lceil\alpha\rceil n\geq\frac{\lceil\alpha\rceil n(\lceil\alpha\rceil(n-m)+1)-\alpha\lceil\alpha\rceil n}{\lceil\alpha\rceil n},
	\]
	which is equivalent to
	\[
	\alpha\geq\lceil\alpha\rceil(n-m)+1.
	\]
	Since $n>m$, this is a contradiction, i.e., $f^{q}\notin\augfuncs_{\alpha}$.
	
	Now, with $X=\{c,b_{1},...,b_{\lceil\alpha\rceil m-1}\}$ and $Y=B\setminus X=\{b_{\lceil\alpha\rceil m},...,b_{\lceil\alpha\rceil n}\}$,
	we have
	\[
	\frac{\sum_{y\in Y}f^{q}(X\cup\{y\})-f^{q}(X)}{f^{q}(X\cup Y)-f^{q}(X)}=0.
	\]
	Thus, and because the set $X$ can be the greedy solution $\greedyset_{\lastk}$,
	the weak submodularity ratio of this problem is $\gamma(f^{q})=0$,
	i.e., $f^{q}\notin\srfuncs_{\gamma}$.
\end{proof}

\section{$\gamma$-$\alpha$-Augmentability\label{sec:gam-alph-aug}}

In this section, we argue that the class $\gafuncs_{\gamma,\alpha}$
of weakly $\gamma$-$\alpha$-augmentable functions unifies and generalizes
the classes $\srfuncs_{\gamma}$, $\augfuncs_{\alpha}$, and $\rqfuncs_{q}$.
We start by proving the first half of Theorem~\ref{thm:gam-alph-aug_weaker_than_rest}.
The second half will be shown in Section~\ref{subsec:lower_bounds},
together with lower bounds for the approximation ratio of the greedy
algorithm.

We need the following simple lemma.
\begin{lem}
	\label{lem:adding_element_to_greedy_set}Let $(\U,\I)$ be an independence
	system with weight function\linebreak \mbox{$w\colon\U\rightarrow\R_{\geq0}$} and
	weighted rank function $f$. Furthermore, let $k\in\oneto{\lastk}$
	and $x\in\U$ with $w(x)>0$. Then, the following are equivalent:
	
	(i) $\greedyset_{k}\cup\{x\}\in\I$
	
	(ii) $f(\greedyset_{k}\cup\{x\})-f(\greedyset_{k})=w(x)$
	
	(iii) $f(\greedyset_{k}\cup\{x\})-f(\greedyset_{k})>0$
\end{lem}

\begin{proof}
	``\textit{(i) $\Rightarrow$ (ii)}'': By definition of $f$ as a
	weighted rank function and because~$\greedyset_{k}\cup\{x\}\in\I$,
	we have
	\[
	f(\greedyset_{k}\cup\{x\})-f(\greedyset_{k})=\sum_{x'\in\greedyset_{k}\cup\{x\}}w(x')-\sum_{x'\in\greedyset_{k}}w(x')=w(x).
	\]
	
	``\textit{(ii) $\Rightarrow$ (iii)}'': This follows immediately
	from the fact that $w(x)>0$.
	
	``\textit{(iii) $\Rightarrow$ (i)}'': Let $x\in\U$ with $f(\greedyset_{k}\cup\{x\})-f(\greedyset_{k})>0$.
	Suppose there is some $s'\in\greedyset_{k}$ with $w(x)>w(s')$. This
	means that $x$ was considered by the greedy algorithm before and
	not added to the solution, i.e., \mbox{$\{s\in\greedyset_{k}\mid w(s)\geq w(x)\}\cup\{x\}\notin\I$}.
	The fact that \mbox{$f(\greedyset_{k}\cup\{x\})-f(\greedyset_{k})>0$}
	implies that there is \mbox{$\emptyset\neq S\subseteq\greedyset_{k}$}
	with\linebreak \mbox{$\greedyset_{k}\setminus S\cup\{x\}\in\I$}
	and $w(S)<w(x)$. The last inequality implies that\linebreak \mbox{$\{s\in S\mid w(s)\geq w(x)\}=\emptyset$},
	which means that \mbox{$\{s\in\greedyset_{k}\mid w(s)\geq w(x)\}\subseteq\greedyset_{k}\setminus S$}.
	But then \mbox{$\{s\in\greedyset_{k}\mid w(s)\geq w(x)\}\cup\{x\}\in\I$},
	which is a contradiction. Therefore, we have \mbox{$w(x)\leq w(s)$}
	for all $s\in\greedyset_{k}$. If we would have \mbox{$\greedyset_{k}\cup\{x\}\notin\I$},
	then the equality \mbox{$f(\greedyset_{k}\cup\{x\})-f(\greedyset_{k})=0$}
	would hold because every element in $\greedyset_{k}$ has a greater
	weight than $x$ and because $\greedyset_{k}\in\I$. Thus, the statement
	holds.
\end{proof}

Since (weak) $\gamma$-$\alpha$-augmentability implies (weak) $\gamma'$-$\alpha'$-augmentability
for all\linebreak \mbox{$\gamma\geq\gamma'$} and $\alpha\leq\alpha'$, the following
proposition implies the first part of Theorem~\ref{thm:gam-alph-aug_weaker_than_rest}.
\begin{prop}
	\label{prop:gamma-alpha_unifies_alpha_and_submod-ratio}For every
	$\gamma,q\in(0,1]$, and every $\alpha\geq1$, it holds that
	\[
	\gafuncs_{1,\alpha}\supseteq\augfuncs_{\alpha}\qquad\mathrm{and}\qquad\gafuncs_{\gamma,\gamma}\supseteq\srfuncs_{\gamma}\qquad\mathrm{and}\qquad\gafuncs_{\gamma,\gamma/q}\supseteq\rqfuncs_{q}.
	\]
\end{prop}

\begin{proof}
	If $f\in\augfuncs_{\alpha}$, then, for all $X,Y\subseteq\U$, and,
	in particular $X\in\{\greedyset_{0},...,\greedyset_{\lastk}\}$, there
	exists $y\in Y$ with
	\[
	f(X\cup\{y\})-f(X)\geq\frac{1\cdot f(X\cup Y)-\alpha f(X)}{|Y|},
	\]
	which means that $f\in\gafuncs_{1,\alpha}$.
	
	For the second part of the proof, let $f\in\srfuncs_{\gamma}$, $X\in\{\greedyset_{0},...,\greedyset_{\lastk}\}$
	and $Y\subseteq\U\setminus X$. Furthermore, let $y^{*}\in\arg\max_{y\in Y}f(X\cup\{y\})$.
	Then, we have
	\begin{align*}
		|Y|(f(X\cup\{y^{*}\})-f(X)) & \geq\sum_{y\in Y}(f(X\cup\{y\})-f(X))\\
		& \geq\gamma(f)f(X\cup Y)-\gamma(f)f(X),
	\end{align*}
	where the second inequality follows from the definition of the weak
	submodularity ratio. Since $\gamma(f)\geq\gamma$, this means that
	$f$ is weakly $\gamma$-$\gamma$-augmentable, i.e., $f\in\gafuncs_{\gamma,\gamma}$.
	
	For the last part of the proof, let $f\in\rqfuncs_{q}$ be the weighted
	rank function of an independence system $(\U,\I)$, and let $w\colon\U\rightarrow\R_{\geq0}$
	be the associated weight function. Furthermore, let $k\in\oneto{\lastk}$
	and $Y\subseteq\U$. We prove that there exists $y\in Y$ with
	\begin{equation}
		f(\greedyset_{k}\cup\{y\})-f(\greedyset_{k})\geq\frac{f(\greedyset_{k}\cup Y)-\frac{1}{q(\U,\I)}f(\greedyset_{k})}{|Y|}.\label{eq:ind-sys_gam-alph-aug}
	\end{equation}
	
	Let $S'\subseteq\greedyset_{k}$ and $Y'\subseteq Y$ with $S'\cup Y'\in\I$
	and $f(\greedyset_{k}\cup Y)=w(S'\cup Y')$. Furthermore, let $y^{*}:=\arg\max_{y\in Y'}f(\greedyset_{k}\cup\{y\})$.
	We define
	\[
	\tilde{Y}:=\begin{cases}
		\{y\in Y'\mid w(y)>w(y^{*})\}, & \textrm{if }f(\greedyset_{k}\cup\{y^{*}\})>f(\greedyset_{k}),\\
		Y', & \textrm{if }f(\greedyset_{k}\cup\{y^{*}\})=f(\greedyset_{k}),
	\end{cases}
	\]
	and we define the independence system $(\tilde{\U},\tilde{\I})$ with
	\begin{eqnarray*}
		\tilde{\U} & := & \greedyset_{k}\cup\tilde{Y},\\
		\tilde{\I} & := & 2^{\greedyset_{k}}\cup2^{S'\cup\tilde{Y}}.
	\end{eqnarray*}
	We have $\tilde{\U}\subseteq\U$ and $\tilde{\I}\subseteq\I$ and
	thus $q(\tilde{\U},\tilde{\I})\geq q(\U,\I)$. The greedy solution
	for the maximization problem on the independence system $(\tilde{\U},\tilde{\I})$
	is $\greedyset_{k}$. Let $\optset\subseteq\tilde{\U}$ be the optimal
	solution. Then, as shown in \cite{Jenkyns1976,Korte1978}, we have
	\begin{equation}
		f(\greedyset_{k})\geq q(\tilde{\U},\tilde{\I})f(\optset)\geq q(\tilde{\U},\tilde{\I})f(S'\cup\tilde{Y})\geq q(\U,\I)f(S'\cup\tilde{Y}).\label{eq:estimate_from_restricted_ind_sys}
	\end{equation}
	If $f(\greedyset_{k}\cup\{y^{*}\})>f(\greedyset_{k})$, Lemma \ref{lem:adding_element_to_greedy_set}
	yields $f(\greedyset_{k}\cup\{y^{*}\})-f(\greedyset_{k})=w(y^{*})$,
	and if $f(\greedyset_{k}\cup\{y^{*}\})=f(\greedyset_{k})$, by definition
	of $\tilde{Y}$, we have $|Y\setminus Y'|=0$. Using this and the
	definition of $\tilde{Y}$, we get
	\begin{eqnarray*}
		|Y|(f(\greedyset_{k}\cup\{y^{*}\})-f(\greedyset_{k})) & \geq & |Y'\setminus\tilde{Y}|w(y^{*})\\
		& \geq & w(Y'\setminus\tilde{Y})\\
		& \overset{\eqref{eq:estimate_from_restricted_ind_sys}}{\geq} & w(Y'\setminus\tilde{Y})+w(S'\cup\tilde{Y})-\frac{1}{q(\U,\I)}f(\greedyset_{k})\\
		& = & w(S'\cup Y')-\frac{1}{q(\U,\I)}f(\greedyset_{k})\\
		& = & f(\greedyset_{k}\cup Y)-\frac{1}{q(\U,\I)}f(\greedyset_{k})\\
		& \overset{1\geq\gamma}{\geq} & \gamma f(\greedyset_{k}\cup Y)-\frac{\gamma}{q(\U,\I)}f(\greedyset_{k}).
	\end{eqnarray*}
	Since $q(\U,\I)\geq q$, this yields weak $\gamma$-$\frac{\gamma}{q}$-augmentability,
	i.e., $f\in\gafuncs_{\gamma,\gamma/q}$.
\end{proof}

Having shown that $\gafuncs_{\gamma,\alpha}$ subsumes the other
three classes of functions, we now prove the upper bound of Theorem~\ref{thm:tight_bound}
for this class. Observe that the upper bound trivially carries over
to the class of monotone, $\gamma$-$\alpha$-augmentable (not weakly)
functions.
\begin{thm}
	\label{thm:gam-alph-aug_upper_bound}The approximation ratio of the
	greedy algorithm on the class $\gafuncs_{\gamma,\alpha}$ of monotone,
	weakly $\gamma$-$\alpha$-augmentable functions, with $\gamma\in(0,1]$
	and $\alpha\geq\gamma$, is at most
	\[
	\frac{\alpha}{\gamma}\cdot\frac{\mathrm{e}^{\alpha}}{\mathrm{e}^{\alpha}-1}.
	\]
\end{thm}

\begin{proof}
	Let $f\in\gafuncs_{\gamma,\alpha}$. First we consider the case $k>\lastk$,
	i.e., the case that the greedy algorithm stops early because the value
	of the solution cannot be increased by adding any element. Because
	$f$ is weakly $\gamma$-$\alpha$-augmentable and there is no element
	$u\in\U$ with $\mbox{\ensuremath{f(S_{\lastk}\cup\{u\})-f(S_{\lastk})>0}}$,
	we have, for all $x\in\optset_{k}$,
	\[
	0=|\optset_{k}|(f(\greedyset_{\lastk}\cup\{x\})-f(\greedyset_{\lastk}))\geq\gamma f(\greedyset_{\lastk}\cup\optset_{k})-\alpha f(\greedyset_{\lastk})\geq\gamma f(\optset_{k})-\alpha f(\greedyset_{k}),
	\]
	i.e., $f(\greedyset_{k})\geq\frac{\gamma}{\alpha}f(\optset_{k})>\frac{\gamma}{\alpha}\cdot\frac{\e^{\alpha}-1}{\e^{\alpha}}f(\optset_{k})$.
	
	Now consider the case that $k\leq\lastk$. For ease of notation, we
	define the gain of the greedy algorithm in iteration $j$ to be \mbox{$\delta_{j}:=f(\greedyset_{j})-f(\greedyset_{j-1})$}
	for all $j\in\oneto k$. Let $i\in\oneto k$ and
	\[
	\mbox{\ensuremath{x^{*}:=\arg\max_{x\in\optset_{k}}f(\greedyset_{i-1}\cup\{x\})-f(\greedyset_{i-1})}}.
	\]
	We have
	\begin{eqnarray}
		\delta_{i} & \geq & f(\greedyset_{i-1}\cup\{x^{*}\})-f(\greedyset_{i-1})\geq\frac{\gamma f(\greedyset_{i-1}\cup\optset_{k})-\alpha f(\greedyset_{i-1})}{|\optset_{k}|}\nonumber \\
		& \geq & \frac{\gamma}{k}f(\optset_{k})-\frac{\alpha}{k}f(\greedyset_{i-1})=\frac{\gamma}{k}f(\optset_{k})-\frac{\alpha}{k}\sum_{j=1}^{i-1}\delta_{j}-\frac{\alpha}{k}f(\emptyset).\label{eq:estimate_gain_greedy_single_step}
	\end{eqnarray}
	We prove by induction that, for all $\ell\in\{0,...,k\}$, we have
	\begin{equation}
		f(\optset_{k})-\frac{\alpha}{\gamma}\sum_{j=1}^{\ell}\delta_{j}-\frac{\alpha}{\gamma}f(\emptyset)\leq f(\optset_{k})(1-\frac{\alpha}{k})^{\ell}.\label{eq:induction_statement_greedy_estimate}
	\end{equation}
	For $\ell=0$ the equation obviously holds. Now suppose that (\ref{eq:induction_statement_greedy_estimate})
	holds for some $\ell\in\{0,...,k-1\}$. Then, for~$\ell+1$, we have
	\begin{eqnarray*}
		& & f(\optset_{k})-\frac{\alpha}{\gamma}\sum_{j=1}^{\ell+1}\delta_{j}-\frac{\alpha}{\gamma}f(\emptyset) \\
		& = & f(\optset_{k})-\frac{\alpha}{\gamma}\sum_{j=1}^{\ell}\delta_{j}-\frac{\alpha}{\gamma}\delta_{\ell+1}-\frac{\alpha}{\gamma}f(\emptyset)\\
		& \overset{\eqref{eq:estimate_gain_greedy_single_step}}{\leq} & f(\optset_{k})-\frac{\alpha}{\gamma}\sum_{j=1}^{\ell}\delta_{j}-\frac{\alpha}{\gamma}\Bigl(\frac{\gamma}{k}f(\optset_{k})-\frac{\alpha}{k}\sum_{j=1}^{\ell}\delta_{j}-\frac{\alpha}{k}f(\emptyset)\Bigr)-\frac{\alpha}{\gamma}f(\emptyset)\\
		& = & \Bigl(f(\optset_{k})-\frac{\alpha}{\gamma}\sum_{j=1}^{\ell}\delta_{j}-\frac{\alpha}{\gamma}f(\emptyset)\Bigr)(1-\frac{\alpha}{k})\overset{\eqref{eq:induction_statement_greedy_estimate}}{\leq}f(\optset_{k})(1-\frac{\alpha}{k})^{\ell+1},
	\end{eqnarray*}
	and (\ref{eq:induction_statement_greedy_estimate}) continues to hold.
	Because of $1+x\leq\e^{x}$ for $x\in\R$, we have
	\[
	f(\optset_{k})-\frac{\alpha}{\gamma}\sum_{j=1}^{\ell}\delta_{j}-\frac{\alpha}{\gamma}f(\emptyset)\overset{\eqref{eq:induction_statement_greedy_estimate}}{\leq}f(\optset_{k})(1-\frac{\alpha}{k})^{\ell}\leq\e^{-\frac{\alpha}{k}\ell}f(\optset_{k}).
	\]
	Rearranging this for $\ell=k$ and using the fact that $f(\greedyset_{k})=\sum_{j=1}^{k}\delta_{j}+f(\emptyset)$,
	yields\setlength{\belowdisplayskip}{-12pt}
	\[
	f(\greedyset_{k})\geq\frac{\gamma}{\alpha}\cdot\frac{\e^{\alpha}-1}{\e^{\alpha}}f(\optset_{k}).
	\]
\end{proof}

\subsection{A Critical Function\label{subsec:lower_bounds}}

To obtain the tight lower bound of Theorem~\ref{thm:tight_bound}
for weakly $\gamma$-$\alpha$-augmentable problems and to separate
this class from $\srfuncs_{\gamma}\cup\augfuncs_{\alpha}\cup\rqfuncs_{q}$,
we introduce a function that is inspired by a construction in~\cite{Bian2017}
for the submodularity ratio.

We fix $\gamma\in(0,1]$ and $\alpha\geq\gamma$. Let $k\in\N$ with
$k>\alpha$, and let $A=\{a_{1},...,a_{k}\}$ and $B=\{b_{1},...,b_{k}\}$
be disjoint sets. We set $\U=A\cup B$, define $\xi_{i}:=\frac{1}{k}(\frac{k-\alpha}{k})^{i-1}$
and let $h(x):=\frac{\gamma^{-1}-1}{k-1}x^{2}+\frac{k-\gamma^{-1}}{k-1}x$.
For our purpose, the important facts about~$h$ are $h(0)=0$, $h(1)=1$,
$h(k)=\frac{k}{\gamma}$ and that $h$ is convex and non-decreasing
on $[0,k])$. With this in mind, we define the function $F_{\gamma,\alpha,k}\colon2^{\U}\rightarrow\R_{\geq0}$
by
\begin{equation*}
	F_{\gamma,\alpha,k}(X) = \max_{X'\subseteq X}\Bigl\{\frac{h(|\{b_{1}\}\cap X'|\cdot|B\cap X'|)}{k}\Bigl(1-\alpha\sum_{\substack{i\in\oneto k:\\
			a_{i}\in A\cap X'
		}
	}\xi_{i}\Bigr)+\sum_{\substack{i\in\oneto k:\\
			a_{i}\in A\cap X'
		}
	}\xi_{i}\Bigr\}
\end{equation*}
If $h(|\{b_{1}\}\cap X|\cdot|B\cap X|)>\frac{k}{\alpha}$, we have
\[
F_{\gamma,\alpha,k}(X) = \frac{h(|\{b_{1}\}\cap X|\cdot|B\cap X|)}{k},
\]
and otherwise, if $h(|\{b_{1}\}\cap X|\cdot|B\cap X|)\leq\frac{k}{\alpha}$, we have
\[
F_{\gamma,\alpha,k}(X) = \frac{h(|\{b_{1}\}\cap X|\cdot|B\cap X|)}{k}\Bigl(1-\alpha\sum_{\substack{i\in\oneto k:\\a_{i}\in A\cap X}}\xi_{i}\Bigr)+\sum_{\substack{i\in\oneto k:\\a_{i}\in A\cap X}}\xi_{i}.
\]

We observe that, for $X\subseteq B$, convexity of $h$, $h(0)=0$,
$h(k)=k/\gamma$ and $|X|\leq|B|=k$ imply that
\begin{equation}
	h(|\{b_{1}\}\cap X|\cdot|X|)\leq\frac{|\{b_{1}\}\cap X|\cdot|X|}{\gamma},\label{eq:h(X)_smaller_divided_by_gamma}
\end{equation}
and, for $\ell\in\{0,...,k\}$, we have 
\begin{equation}
	\sum_{i=1}^{\ell}\xi_{i}=\sum_{i=1}^{\ell}\frac{1}{k}\Bigl(\frac{k-\alpha}{k}\Bigr)^{i-1}=\frac{1}{k}\cdot\frac{1-\bigl(\frac{k-\alpha}{k}\bigr)^{\ell}}{1-\frac{k-\alpha}{k}}=\frac{1-\bigl(\frac{k-\alpha}{k}\bigr)^{\ell}}{\alpha}.\label{eq:sum_xi_1_to_ell}
\end{equation}
We show that our modification of the function introduced in~\cite{Bian2017}
retains the same structure in regard to greedy solutions.
\begin{prop}
	\label{prop:greedy_pickorder_F}For $i\in\oneto k$, the greedy algorithm
	picks the element $a_{i}$ in iteration $i$, and, for $\mbox{\ensuremath{i\in\oneto{2k}\setminus\oneto k}}$,
	the greedy algorithm picks the element $b_{i-k}$ in iteration $i$.
\end{prop}

\begin{proof}
	First, we consider the case $i\in\oneto k$. Suppose that in iteration
	$i$, the initial solution is $\{a_{1},...,a_{i-1}\}$, where $\{a_{1},...,a_{0}\}=\emptyset$,
	with objective value $\sum_{\ell=1}^{i-1}\xi_{\ell}$. Adding an element
	from $\{b_{2},...,b_{k}\}$ does not increase the objective value
	because $\{b_{1}\}\cap\{b\}=\emptyset$ for all $b\in\{b_{2},...,b_{k}\}$.
	For $j\in\{i,...,k\}$, adding $a_{j}$ increases the objective value
	by $\xi_{j}=\frac{1}{k}(\frac{k-\alpha}{k})^{j-1}$. Since $k>\alpha$,
	we have $\xi_{i}\geq\xi_{j}$ for $j\geq i$. Adding the element $b_{1}$
	to the solution $\{a_{1},...,a_{i-1}\}$ increases the objective value
	by 
	\[
	\frac{1}{k}\Bigl(1-\alpha\sum_{\ell=1}^{i-1}\xi_{\ell}\Bigr)\overset{\eqref{eq:sum_xi_1_to_ell}}{=}\frac{1}{k}\Bigl(1-\alpha\frac{1-\bigl(\frac{k-\alpha}{k}\bigr)^{i-1}}{\alpha}\Bigr)=\frac{1}{k}\Bigl(\frac{k-\alpha}{k}\Bigr)^{i-1}.
	\]
	Thus, with proper tie breaking, the greedy algorithm picks the element
	$a_{i}$ in iteration $i$ for $i\in\oneto k$.
	
	Now, we consider the case that $i\in\{k+1,...,2k\}$. For $i=k+1$,
	adding an element from $\{b_{2},...,b_{k}\}$ does not increase the
	objective value, while adding $b_{1}$ increases it by~$\frac{1}{k}(\frac{k-\alpha}{k})^{k}$.
	Thus, in iteration $k+1$, the element~$b_{1}$ is added to the solution.
	For~\mbox{$i\geq k+2$}, adding any element from $B\setminus\greedyset_{i-1}$
	to the greedy solution $\greedyset_{i-1}$ increases the function
	value by the same amount. Therefore, with proper tie breaking, the
	greedy algorithm picks the element $b_{i-k}$ in iteration $i$ for
	$i\in\{k+1,...,2k\}$.
\end{proof}

With this, we can show that $F_{\gamma,\alpha,k}$ is weakly $\gamma$-$\alpha$-augmentable.
\begin{lem}
	\label{lem:F_monotone_and_gam-alph-aug}For every $\gamma\in(0,1]$,
	every $\alpha\geq\gamma$, and every $k\in\N$ with $k>\alpha$, it
	holds that $\mbox{\ensuremath{F_{\gamma,\alpha,k}\in\gafuncs_{\gamma,\alpha}}}$.
\end{lem}

\begin{proof}
	The monotonicity of $F_{\gamma,\alpha,k}$ immediately follows from
	the maximum in the definition. To prove weak $\gamma$-$\alpha$-augmentability,
	let $X\in\{\greedyset_{0},...,\greedyset_{2k}\}$ and $Y\subseteq U$.
	We define $Y':=Y\setminus X$. For better readability, we will write
	$F:=F_{\gamma,\alpha,k}$.
	
	First, consider the case that $X\subseteq A$. Then $F(X)=\sum_{i\in\oneto k:a_{i}\in X}\xi_{i}$
	because $h(0)=0$. Thus and because $h(1)=1$, for all $y\in Y'$,
	we have
	\[
	F(X\cup\{y\})-F(X)=\begin{cases}
		\xi_{i}, & \textrm{if }y=a_{i}\in(A\cap Y'),\\
		\frac{1}{k}(1-\alpha\sum_{i\in\oneto k:a_{i}\in X}\xi_{i}), & \textrm{if }y=b_{1},\\
		0, & \textrm{else},
	\end{cases}
	\]
	i.e.,
	\begin{eqnarray}
		& & |Y'|\bigl(\max_{y\in Y'}F(X\cup\{y\})-F(X)\bigr)\nonumber \\
		& \geq & \Bigl(\sum_{y\in A\cap Y'}\bigl(F(X\cup\{y\})-F(X)\bigr)\Bigr)+\nonumber \\
		&  & |B\cap Y'|\bigl(\max_{y\in B\cap Y'}F(X\cup\{y\})-F(X)\bigr)\nonumber \\
		& = & \Bigl(\sum_{\substack{i\in\oneto k:\\
				a_{i}\in A\cap Y'
			}
		}\xi_{i}\Bigr)+|\{b_{1}\}\cap Y'|\cdot|B\cap Y'|\frac{1}{k}\Bigl(1-\alpha\sum_{\substack{i\in\oneto k:\\
				a_{i}\in X
			}
		}\xi_{i}\Bigr).\label{eq:gam-alph-aug_of_F_first_part}
	\end{eqnarray}
	If $h(|\{b_{1}\}\cap Y'|\cdot|B\cap Y'|)\leq\frac{k}{\alpha}$, we
	use the fact that $F(X)=\sum_{i\in\oneto k:a_{i}\in X}\xi_{i}$ to
	calculate
	\begin{eqnarray}
		&  & \gamma F(X\cup Y)-\alpha F(X)\nonumber \\
		& \overset{B\cap X=\emptyset}{=} & \gamma\Bigl(\frac{h(|\{b_{1}\}\cap Y'|\cdot|B\cap Y'|)}{k}\Bigl(1-\alpha\sum_{\substack{i\in\oneto k:\\
				a_{i}\in X\cup(A\cap Y')
			}
		}\xi_{i}\Bigr)\nonumber \\
		& & +\sum_{\substack{i\in\oneto k:\\
				a_{i}\in X\cup(A\cap Y')
			}
		}\xi_{i}\Bigr)-\alpha\sum_{\substack{i\in\oneto k:\\
				a_{i}\in X
			}
		}\xi_{i}\nonumber \\
		& = & \Bigl[\frac{\gamma}{k}h(|\{b_{1}\}\cap Y'|\cdot|B\cap Y'|)\Bigl(1-\alpha\sum_{\substack{i\in\oneto k:\\
				a_{i}\in X
			}
		}\xi_{i}\Bigr)\Bigr]\nonumber \\
		&  & +\Bigl[\gamma\bigl(1-\frac{\alpha}{k}h(|\{b_{1}\}\cap Y'|\cdot|B\cap Y'|)\bigr)\Bigl(\sum_{\substack{i\in\oneto k:\\
				a_{i}\in A\cap Y'
			}
		}\xi_{i}\Bigr)\Bigr]+\Bigl[(\gamma-\alpha)\sum_{\substack{i\in\oneto k:\\
				a_{i}\in X
			}
		}\xi_{i}\Bigr]\nonumber \\
		& \leq & \Bigl[\frac{1}{k}|\{b_{1}\}\cap Y'|\cdot|B\cap Y'|\Bigl(1-\alpha\sum_{\substack{i\in\oneto k:\\
				a_{i}\in X
			}
		}\xi_{i}\Bigr)\Bigr]+\Bigl[\sum_{\substack{i\in\oneto k:\\
				a_{i}\in A\cap Y'
			}
		}\xi_{i}\Bigr]+[0].\label{eq:gam-alph-aug_of_F_second_part}
	\end{eqnarray}
	The first part of the last inequality follows from (\ref{eq:h(X)_smaller_divided_by_gamma}).
	The second part of the inequality follows from the fact that $\gamma\in(0,1]$
	and, for $x\geq0$, we have $\frac{\alpha}{k}h(x)\geq0$. The last
	part follows from the fact that $\gamma\leq\alpha$. Combining equations
	(\ref{eq:gam-alph-aug_of_F_first_part}) and (\ref{eq:gam-alph-aug_of_F_second_part})
	together with the fact that $Y'\subseteq Y$ yields weak $\gamma$-$\alpha$-augmentability.
	
	Otherwise, if $h(|\{b_{1}\}\cap Y'|\cdot|B\cap Y'|)>\frac{k}{\alpha}$,
	we have
	\begin{eqnarray*}
		\gamma F(X\cup Y)-\alpha F(X) & = & \gamma\frac{h(|\{b_{1}\}\cap Y'|\cdot|B\cap Y'|)}{k}-\alpha\sum_{\substack{i\in\oneto k:\\
				a_{i}\in X
			}
		}\xi_{i}\\
		& \overset{\eqref{eq:h(X)_smaller_divided_by_gamma}}{\leq} & \frac{1}{k}|\{b_{1}\}\cap Y'|\cdot|B\cap Y'|-\alpha\sum_{\substack{i\in\oneto k:\\
				a_{i}\in X
			}
		}\xi_{i}\\
		& \overset{|B|=k}{\leq} & \frac{1}{k}|\{b_{1}\}\cap Y'|\cdot|B\cap Y'|\Bigl(1-\alpha\sum_{\substack{i\in\oneto k:\\
				a_{i}\in X
			}
		}\xi_{i}\Bigr)\\
		& \overset{\eqref{eq:gam-alph-aug_of_F_first_part}}{\leq} & |Y'|\bigl(\max_{y\in Y'}F(X\cup\{y\})-F(X)\bigr),
	\end{eqnarray*}
	which yields $\gamma$-$\alpha$-augmentability also in this case.
	
	Now, consider the case that $X\nsubseteq A$. Then, by Proposition
	\ref{prop:greedy_pickorder_F}, we have\linebreak \mbox{$X=A\cup\{b_{1},...,b_{i}\}$}
	for some~$i\in\oneto k$. If $i=k$, i.e., $X=U$, we have\linebreak \mbox{$\gamma F(X\cup Y)-\alpha F(X)=(\gamma-\alpha)F(U)\leq0$}
	and we are done. Thus, assume that $i<k$. Because $h$ is convex,
	we have
	\begin{equation}
		\frac{h(i+|Y'|)-h(i)}{|Y'|}\leq\frac{h(|B|)-h(|B|-|Y'|)}{|B|-i}\overset{|Y'|\leq|B|-i}{\leq}\frac{h(|B|)-h(i)}{|B|-i}.\label{eq:h_estimate_part_1}
	\end{equation}
	With
	\[
	H(i):=(k-i)\frac{h(i+1)-h(i)}{h(k)-h(i)},
	\]
	we have
	\[
	H'(i)=(k-1)\frac{2-3\gamma+\gamma^{2}}{(k-1+i-\gamma i)^{2}}\geq0,
	\]
	which yields
	\begin{equation}
		H(i)\geq H(0)=k\frac{1-0}{\frac{k}{\gamma}-0}=\gamma.\label{eq:h_estimate_part_2}
	\end{equation}
	Combining this with (\ref{eq:h_estimate_part_1}), we obtain
	\begin{equation}
		\frac{|Y'|\bigl(h(i+1)-h(i)\bigr)}{h(i+|Y'|)-h(i)}\overset{\eqref{eq:h_estimate_part_1}}{\geq}\frac{(|B|-i)\bigl(h(i+1)-h(i)\bigr)}{h(|B|)-h(i)}\overset{|B|=k}{=}H(i)\overset{\eqref{eq:h_estimate_part_2}}{\geq}\gamma.\label{eq:h_submodularity_ratio}
	\end{equation}
	If $h(i+|Y'|)\leq\frac{k}{\alpha}$, because $h$ is increasing for
	positive values, we have\linebreak \mbox{$h(i)\leq h(i+1)\leq h(i+|Y'|)\leq\frac{k}{\alpha}$}.
	Thus, for every $y\in Y'$, we have
	\begin{eqnarray}
		|Y|\bigl(F(X\cup\{y\})-F(X)\bigr) & \geq & |Y'|\frac{h(i+1)-h(i)}{k}\Bigl(1-\alpha\sum_{j=1}^{k}\xi_{j}\Bigr)\nonumber \\
		& \overset{\eqref{eq:h_submodularity_ratio}}{\geq} & \frac{\gamma\bigl(h(i+|Y'|)-h(i)\bigr)}{k}\Bigl(1-\alpha\sum_{j=1}^{k}\xi_{j}\Bigr)\nonumber \\
		& \overset{\gamma\leq\alpha}{\geq} & \Bigl(\gamma\frac{h(i+|Y'|)}{k}-\alpha\frac{h(i)}{k}\Bigr)\Bigl(1-\alpha\sum_{j=1}^{k}\xi_{j}\Bigr)\label{eq:F_gam-alph-aug_intermediate_step}\\
		& \overset{\gamma\leq\alpha}{\geq} & \gamma\Bigl[\frac{h(i+|Y'|)}{k}\Bigl(1-\alpha\sum_{j=1}^{k}\xi_{j}\Bigr)+\sum_{j=1}^{k}\xi_{j}\Bigr]\nonumber \\
		&  & -\alpha\Bigl[\frac{h(i)}{k}\Bigl(1-\alpha\sum_{j=1}^{k}\xi_{j}\Bigr)+\sum_{j=1}^{k}\xi_{j}\Bigr]\nonumber \\
		& = & \gamma F(X\cup Y)-\alpha F(X).\nonumber 
	\end{eqnarray}
	If $h(i)\leq h(i+1)\leq\frac{k}{\alpha}<h(i+|Y'|)$, then, for every
	$y\in Y'$, we have
	\begin{eqnarray}
		&  & |Y|\bigl(F(X\cup\{y\})-F(X)\bigr)\nonumber \\
		& \overset{\eqref{eq:F_gam-alph-aug_intermediate_step}}{\geq} & \Bigl(\gamma\frac{h(i+|Y'|)}{k}-\alpha\frac{h(i)}{k}\Bigr)\Bigl(1-\alpha\sum_{j=1}^{k}\xi_{j}\Bigr)\nonumber \\
		& = & \gamma\frac{h(i+|Y'|)}{k}-\alpha\frac{h(i)}{k}\Bigl(1-\alpha\sum_{j=1}^{k}\xi_{j}\Bigr)-\frac{\gamma}{k}h(i+|Y'|)\alpha\sum_{j=1}^{k}\xi_{j}\nonumber \\
		& \overset{h(i+|Y'|)\leq h(k)=\frac{k}{\gamma}}{\geq} & \gamma\frac{h(i+|Y'|)}{k}-\alpha\Bigl[\frac{h(i)}{k}\Bigl(1-\alpha\sum_{j=1}^{k}\xi_{j}\Bigr)+\sum_{j=1}^{k}\xi_{j}\Bigr]\nonumber \\
		& = & \gamma F(X\cup Y)-\alpha F(X).\label{eq:F_gam-alph-aug_case2}
	\end{eqnarray}
	If $h(i)\leq\frac{k}{\alpha}<h(i+1)\leq h(i+|Y'|)$, then
	\begin{equation}
		\frac{\alpha}{k}h(i+1)>\frac{\alpha}{k}\cdot\frac{k}{\alpha}=1,\label{eq:F_gam-alph-aug_case3_estimate}
	\end{equation}
	which implies that, for every $y\in Y'$, we have
	\begin{equation}
		F(X\cup\{y\})=\frac{h(i+1)}{k}\overset{\eqref{eq:F_gam-alph-aug_case3_estimate}}{>}\frac{h(i+1)}{k}\Bigl(1-\alpha\sum_{j=1}^{k}\xi_{j}\Bigr)+\sum_{j=1}^{k}\xi_{j}.\label{eq:F_gam-alph-aug_case3}
	\end{equation}
	This implies
	\begin{eqnarray*}
		|Y|\bigl(F(X\cup\{y\})-F(X)\bigr) & \overset{\eqref{eq:F_gam-alph-aug_case3}}{\geq} & |Y'|\frac{h(i+1)-h(i)}{k}\Bigl(1-\alpha\sum_{j=1}^{k}\xi_{j}\Bigr)\\
		& \overset{\eqref{eq:F_gam-alph-aug_case2}}{\geq} & \gamma F(X\cup Y)-\alpha F(X).
	\end{eqnarray*}
	If $\frac{k}{\alpha}<h(i)\leq h(i+1)\leq h(i+|Y'|)$, then, for every
	$y\in Y'$, we have
	\begin{eqnarray*}
		|Y|\bigl(F(X\cup\{y\})-F(X)\bigr) & = & |Y'|\frac{h(i+1)-h(i)}{k}\\
		& \overset{\eqref{eq:h_submodularity_ratio}}{\geq} & \frac{\gamma\bigl(h(i+|Y'|)-h(i)\bigr)}{k}\\
		& \geq & \gamma\frac{h(i+|Y'|)}{k}-\alpha\frac{h(i)}{k}\\
		& = & \gamma F(X\cup Y)-\alpha F(X),
	\end{eqnarray*}
	i.e., also in all of these cases, $F$ is $\gamma$-$\alpha$-augmentable.
\end{proof}

It is straightforward to bound the approximation ratio of the greedy
algorithm for $F_{\gamma,\alpha,k}$.
\begin{prop}
	\label{prop:approx-ratio_of_F}The approximation ratio of the greedy
	algorithm for maximizing the function $F_{\gamma,\alpha,k}$, with
	$\gamma\in(0,1]$, $\alpha\geq\gamma$ and $k\in\N$ with $k>\alpha$,
	is at least
	\[
	\frac{\alpha}{\gamma}\frac{1}{1-(1-\frac{\alpha}{k})^{k}}.
	\]
\end{prop}

\begin{proof}
	We compare the objective values of the greedy solution $\greedyset_{k}$
	of size $k$ and the solution $B$, which also has size $k$. By Proposition
	\ref{prop:greedy_pickorder_F}, we have $\greedyset_{k}=A$, and thus
	\[
	F(\greedyset_{k})=F(A)=\sum_{i=1}^{k}\xi_{i}\overset{\eqref{eq:sum_xi_1_to_ell}}{=}\frac{1-\bigl(\frac{k-\alpha}{k}\bigr)^{k}}{\alpha}
	\]
	and
	\[
	F(B)=\frac{h(k)}{k}=\frac{\frac{k}{\gamma}}{k}=\frac{1}{\gamma}.
	\]
	Thus, the greedy algorithm has an approximation ratio of at least\setlength{\belowdisplayskip}{-12pt}
	\[
	\frac{F(\optset_{k})}{F(\greedyset_{k})}\geq\frac{F(B)}{F(\greedyset_{k})}=\frac{\alpha}{\gamma}\cdot\frac{1}{1-\bigl(\frac{k-\alpha}{k}\bigr)^{k}}.
	\]
\end{proof}
\begin{thm}
	The approximation ratio of the greedy algorithm on the class $\gafuncs_{\gamma,\alpha}$
	of monotone, weakly $\gamma$-$\alpha$-augmentable functions, with
	$\gamma\in(0,1]$ and $\alpha\geq\gamma$, is at least\label{thm:general_lower_bound}
	\[
	\frac{\alpha}{\gamma}\cdot\frac{\mathrm{e}^{\alpha}}{\mathrm{e}^{\alpha}-1}.
	\]
\end{thm}

\begin{proof}
	By Lemma~\ref{lem:F_monotone_and_gam-alph-aug}, $F_{\gamma,\alpha,k}\in\gafuncs_{\gamma,\alpha}$,
	and, by Proposition~\ref{prop:approx-ratio_of_F}, the greedy algorithm
	has an approximation ratio of at least $\frac{\alpha}{\gamma}\frac{1}{1-(1-\frac{\alpha}{k})^{k}}$
	for maximizing $F_{\gamma,\alpha,k}$. The general lower bound follows,
	since\setlength{\belowdisplayskip}{-12pt}
	\[
	\lim_{k\rightarrow\infty}\frac{1}{1-\bigl(\frac{k-\alpha}{k}\bigr)^{k}}=\frac{1}{1-\e^{-\alpha}}=\frac{\e^{\alpha}}{\e^{\alpha}-1}.
	\]
\end{proof}

It even turns out that, for $\gamma=1$, the function $F_{\gamma,\alpha,k}$
is $\alpha$-augmentable. This allows to carry the lower bound over
to the class $\augfuncs_{\alpha}$.
\begin{prop}
	\label{prop:F_1,alph,k_is_alph-aug}For every $\alpha\geq1$, and
	every $k\in\N$ with $k\geq\alpha$, it holds that~$F_{1,\alpha,k}\in\augfuncs_{\alpha}$.
\end{prop}

\begin{proof}
	By Lemma \ref{lem:F_monotone_and_gam-alph-aug}, $F_{1,\alpha,k}$
	is monotone. Thus, it suffices to prove that the function is $\alpha$-augmentability.
	For better readability, we write $F:=F_{1,\alpha,k}$. Observe that,
	if $\gamma=1$, we have $h(x)=x$ for all $x\in\R$. Let $X,Y\subseteq U$
	and $Y':=Y\setminus X$.
	
	If $|\{b_{1}\}\cap(X\cup Y)|\cdot|B\cap(X\cup Y)|\leq\frac{k}{\alpha}$,
	for $y\in Y'$, we have
	\begin{eqnarray}
		&  & F(X\cup\{y\})-F(X)\nonumber \\
		& = & \begin{cases}
			\bigl(1-\frac{|\{b_{1}\}\cap X|\cdot|B\cap X|}{k}\alpha\bigr)\xi_{i}, & \textrm{if }y=a_{i}\in A\cap Y',\\
			\frac{|\{b_{1}\}\cap(X\cup\{y\})|\cdot|B\cap(X\cup\{y\})|-|\{b_{1}\}\cap X|\cdot|B\cap X|}{k}&\\
			\quad\quad \cdot\bigl(1-\alpha\sum_{i\in\oneto k:a_{i}\in A\cap X}\xi_{i}\bigr), & \textrm{if }y\in B\cap Y'.
		\end{cases}\label{eq:F_alph_aug_adding_one_element_case1}
	\end{eqnarray}
	This yields
	
	\begin{eqnarray*}
		&  & F(X\cup Y)-\alpha F(X)\\
		& = & \Bigl(\frac{|\{b_{1}\}\cap(X\cup Y')|\cdot|B\cap(X\cup Y')|}{k}\Bigl(1-\alpha\sum_{\substack{i\in\oneto k:\\
				a_{i}\in A\cap(X\cup Y')
			}
		}\xi_{i}\Bigr)+\sum_{\substack{i\in\oneto k:\\
				a_{i}\in A\cap(X\cup Y')
			}
		}\xi_{i}\Bigr)\\
		&  & -\alpha\Bigl(\frac{|\{b_{1}\}\cap X|\cdot|B\cap X|}{k}\Bigl(1-\alpha\sum_{\substack{i\in\oneto k:\\
				a_{i}\in A\cap X
			}
		}\xi_{i}\Bigr)+\sum_{\substack{i\in\oneto k:\\
				a_{i}\in A\cap X
			}
		}\xi_{i}\Bigr)\\
		& = & \Bigl[\frac{|\{b_{1}\}\cap(X\cup Y')|\cdot|B\cap(X\cup Y')|-\alpha|\{b_{1}\}\cap X|\cdot|B\cap X|}{k}\Bigl(1-\alpha\sum_{\substack{i\in\oneto k:\\
				a_{i}\in A\cap X
			}
		}\xi_{i}\Bigr)\Bigr]\\
		&  & +\Bigl[\Bigl(1-\frac{|\{b_{1}\}\cap(X\cup Y')|\cdot|B\cap(X\cup Y')|}{k}\alpha\Bigr)\sum_{\substack{i\in\oneto k:\\
				a_{i}\in A\cap Y'
			}
		}\xi_{i}\Bigr]+\Bigl[(1-\alpha)\sum_{\substack{i\in\oneto k:\\
				a_{i}\in A\cap X
			}
		}\xi_{i}\Bigr]\\
		& \leq & \Bigl[|B\cap Y'|\max_{y\in B\cap Y'}\Bigl\{\frac{|\{b_{1}\}\cap(X\cup\{y\})|\cdot|B\cap(X\cup\{y\})|-|\{b_{1}\}\cap X|\cdot|B\cap X|}{k}\\
		&  & \bigl(1-\alpha\sum_{\substack{i\in\oneto k:\\
				a_{i}\in A\cap X
			}
		}\xi_{i}\bigr)\Bigr\}\Bigr]+\Bigl[\Bigl(1-\frac{|\{b_{1}\}\cap X|\cdot|B\cap X|}{k}\alpha\Bigr)\sum_{\substack{i\in\oneto k:\\
				a_{i}\in A\cap Y'
			}
		}\xi_{i}\Bigr]+[0]\\
		& \overset{\eqref{eq:F_alph_aug_adding_one_element_case1}}{=} & |B\cap Y'|\bigl(\max_{y\in B\cap Y'}F(X\cup\{y\})-F(X)\bigr)+\sum_{y\in A\cap Y'}(F(X\cup\{y\})-F(X))\\
		& \leq & |Y|\bigl(\max_{y\in Y}F(X\cup\{y\})-F(X)\bigr).
	\end{eqnarray*}
	This establishes $\alpha$-augmentability if $|\{b_{1}\}\cap(X\cup Y)|\cdot|B\cap(X\cup Y)|\leq\frac{k}{\alpha}$.
	
	Consider the case that $|\{b_{1}\}\cap X|\cdot|B\cap X|\leq\frac{k}{\alpha}<|\{b_{1}\}\cap(X\cup Y)|\cdot|B\cap(X\cup Y)|$.
	If, for $y\in B\cap Y'$, we have $|\{b_{1}\}\cap(X\cup\{y\})|\cdot|B\cap(X\cup\{y\})|\leq\frac{k}{\alpha}$,
	then
	\begin{eqnarray*}
		&  & F(X\cup\{y\})-F(X)\\
		& \overset{\eqref{eq:F_alph_aug_adding_one_element_case1}}{=} & \frac{|\{b_{1}\}\cap(X\cup\{y\})|\cdot|B\cap(X\cup\{y\})|-|\{b_{1}\}\cap X|\cdot|B\cap X|}{k}\bigl(1-\alpha\sum_{i\in\oneto k:a_{i}\in A\cap X}\xi_{i}\bigr)
	\end{eqnarray*}
	and if
	\begin{equation}
		|\{b_{1}\}\cap(X\cup\{y\})|\cdot|B\cap(X\cup\{y\})|>\frac{k}{\alpha}\label{eq:F_alph-aug_adding_one_element_case2_estimate}
	\end{equation}
	for $y\in B\cap Y'$, we have
	\begin{eqnarray*}
		&  & F(X\cup\{y\})-F(X)\\
		& = & \frac{|B\cap(X\cup\{y\})|}{k}-\frac{|\{b_{1}\}\cap X|\cdot|B\cap X|}{k}\bigl(1-\alpha\sum_{i\in\oneto k:a_{i}\in A\cap X}\xi_{i}\bigr)-\sum_{i\in\oneto k:a_{i}\in A\cap X}\xi_{i}\\
		& \overset{\eqref{eq:F_alph-aug_adding_one_element_case2_estimate}}{\geq} & \frac{|\{b_{1}\}\cap(X\cup\{y\})|\cdot|B\cap(X\cup\{y\})|-|\{b_{1}\}\cap X|\cdot|B\cap X|}{k}\bigl(1-\alpha\sum_{i\in\oneto k:a_{i}\in A\cap X}\xi_{i}\bigr).
	\end{eqnarray*}
	This means that, in either case, for $y\in B\cap Y'$, we have
	\begin{eqnarray}
		&  & F(X\cup\{y\})-F(X)\nonumber \\
		& \geq & \frac{|\{b_{1}\}\cap(X\cup\{y\})|\cdot|B\cap(X\cup\{y\})|-|\{b_{1}\}\cap X|\cdot|B\cap X|}{k}\label{eq:F_alph-aug_adding_one_element_case2} \\
		& & \quad\quad \cdot\bigl(1-\alpha\sum_{\substack{i\in\oneto k:\\
				a_{i}\in A\cap X
			}
		}\xi_{i}\bigr).\nonumber
	\end{eqnarray}
	Since we consider $|\{b_{1}\}\cap X|\cdot|B\cap X|\leq\frac{k}{\alpha}<|\{b_{1}\}\cap(X\cup Y)|\cdot|B\cap(X\cup Y)|$,
	we have
	\begin{equation}
		Y'\cap B\neq\emptyset.\label{eq:F_alph-aug_adding_one_element_case3_int-step}
	\end{equation}
	This yields
	
	\begin{eqnarray*}
		&  & F(X\cup Y)-\alpha F(X)\\
		& = & \frac{|B\cap(X\cup Y')|}{k}-\alpha\Bigl(\frac{|\{b_{1}\}\cap X|\cdot|B\cap X|}{k}\Bigl(1-\alpha\sum_{\substack{i\in\oneto k:\\
				a_{i}\in A\cap X
			}
		}\xi_{i}\Bigr)+\sum_{\substack{i\in\oneto k:\\
				a_{i}\in A\cap X
			}
		}\xi_{i}\Bigr)\\
		& = & \Bigl[\frac{|B\cap(X\cup Y')|}{k}-\alpha\sum_{\substack{i\in\oneto k:\\
				a_{i}\in A\cap X
			}
		}\xi_{i}\Bigr]-\Bigl[\alpha\frac{|\{b_{1}\}\cap X|\cdot|B\cap X|}{k}\Bigl(1-\alpha\sum_{\substack{i\in\oneto k:\\
				a_{i}\in A\cap X
			}
		}\xi_{i}\Bigr)\Bigr]\\
		& \overset{\alpha\geq1}{\leq} & \Bigl[\frac{|B\cap(X\cup Y')|}{k}-\alpha\sum_{\substack{i\in\oneto k:\\
				a_{i}\in A\cap X
			}
		}\xi_{i}\Bigr]-\Bigl[\frac{|\{b_{1}\}\cap X|\cdot|B\cap X|}{k}\Bigl(1-\alpha\sum_{\substack{i\in\oneto k:\\
				a_{i}\in A\cap X
			}
		}\xi_{i}\Bigr)\Bigr]\\
		& = & \Bigl[\frac{|B\cap(X\cup Y')|}{k}-\alpha\sum_{\substack{i\in\oneto k:\\
				a_{i}\in A\cap X
			}
		}\xi_{i}\Bigr]\\
		&  & -\Bigl[(|B\cap Y'|-(|B\cap Y'|-1))\frac{|\{b_{1}\}\cap X|\cdot|B\cap X|}{k}\Bigl(1-\alpha\sum_{\substack{i\in\oneto k:\\
				a_{i}\in A\cap X
			}
		}\xi_{i}\Bigr)\Bigr]\\
		& \overset{|B|=k,|\{b_{1}\}\cap X|\leq1}{\leq} & \Bigl[\frac{|B\cap(X\cup Y')|}{k}\bigl(1-\alpha\sum_{\substack{i\in\oneto k:\\
				a_{i}\in A\cap X
			}
		}\xi_{i}\bigr)\Bigr]\\
		&  & -\Bigl[\frac{|B\cap Y'|\cdot|\{b_{1}\}\cap X|\cdot|B\cap X|-(|B\cap Y'|-1)|B\cap X|}{k}\Bigl(1-\alpha\sum_{\substack{i\in\oneto k:\\
				a_{i}\in A\cap X
			}
		}\xi_{i}\Bigr)\Bigr]\\
		& = & |B\cap Y'|\frac{|B\cap X|+1-|\{b_{1}\}\cap X|\cdot|B\cap X|}{k}\bigl(1-\alpha\sum_{\substack{i\in\oneto k:\\
				a_{i}\in A\cap X
			}
		}\xi_{i}\bigr) \\
		& \overset{\eqref{eq:F_alph-aug_adding_one_element_case3_int-step}}{=} & |B\cap Y'|\max_{y\in B\cap Y'}\Bigl\{\frac{|\{b_{1}\}\cap(X\cup\{y\})|\cdot|B\cap(X\cup\{y\})|-|\{b_{1}\}\cap X|\cdot|B\cap X|}{k} \\
		& & \quad\quad\cdot \bigl(1-\alpha\sum_{\substack{i\in\oneto k:\\
				a_{i}\in A\cap X
			}
		}\xi_{i}\bigr)\Bigr\}\\
		& \overset{\eqref{eq:F_alph-aug_adding_one_element_case2}}{\leq} & |B\cap Y'|\bigl(\max_{y\in B\cap Y'}F(X\cup\{y\})-F(X)\bigr)
	\end{eqnarray*}
	Thus, if $|\{b_{1}\}\cap X|\cdot|B\cap X|\leq\frac{k}{\alpha}<|\{b_{1}\}\cap(X\cup Y)|\cdot|B\cap(X\cup Y)|$,the function ~$F$ is $\alpha$-augmentability.
	
	If $\frac{k}{\alpha}<|\{b_{1}\}\cap X|\cdot|B\cap X|$, for $y\in Y'$,
	we have
	\begin{eqnarray}
		F(X\cup\{y\})-F(X) & = & \begin{cases}
			0 & \textrm{if }y=a_{i}\in A\cap Y',\\
			\frac{|B\cap(X\cup\{y\})|-|B\cap X|}{k} & \textrm{if }y\in B\cap Y',
		\end{cases}\label{eq:F_alph_aug_adding_one_element_case3}
	\end{eqnarray}
	which yields
	
	\begin{eqnarray*}
		F(X\cup Y)-\alpha F(X) & = & \frac{|B\cap(X\cup Y')|}{k}-\alpha\frac{|B\cap X|}{k}\\
		& \overset{\alpha\geq1}{\leq} & \frac{|B\cap(X\cup Y')|-|B\cap X|}{k}\\
		& = & \frac{|B\cap Y'|}{k}\\
		& = & |B\cap Y'|\max_{y\in B\cap Y'}\frac{|B\cap(X\cup\{y\})|-|B\cap X|}{k}\\
		& \leq & |Y|\max_{y\in B\cap Y'}\frac{|B\cap(X\cup\{y\})|-|B\cap X|}{k}\\
		& \overset{\eqref{eq:F_alph_aug_adding_one_element_case3}}{=} & |Y|\bigl(\max_{y\in Y}F(X\cup\{y\})-F(X)\bigr).
	\end{eqnarray*}
	This establishes $\alpha$-augmentability if $\frac{k}{\alpha}<|\{b_{1}\}\cap X|\cdot|B\cap X|$.
\end{proof}

Together with Proposition~\ref{prop:approx-ratio_of_F}, this extends
the lower bound of Theorem~\ref{thm:lower_bound_mcflows} to all
$\alpha\geq1$ and thus proves Corollary~\ref{cor:augmentable_lower_bound}.

With this, we can prove the second part of Theorem~\ref{thm:gam-alph-aug_weaker_than_rest}.
\begin{prop}
	\label{prop:F_separates_gam-alph-aug}For every $\gamma'\in(0,1)$,
	$\alpha'\geq\gamma'$, $\alpha\geq1$ and $k\in\N$ with $k>\alpha'$,
	it holds that $F_{\gamma',\alpha',k}\notin\augfuncs_{\alpha}$. For
	every $\gamma,\gamma',q\in(0,1]$ and $\alpha'\geq\gamma'$, there
	exists~$k'\in\N$ with $k'>\alpha$ such that $F_{\gamma',\alpha',k'}\nsubseteq\srfuncs_{\gamma}\cup\rqfuncs_{q}$.
\end{prop}

\begin{proof}
	For the first part, let $\gamma'\in(0,1)$, $\alpha'\geq\gamma'$
	and $k\in\N$ with $k>\alpha'$. Furthermore, let $X=\emptyset$ and
	$Y=B$. For $y\in Y$, we have
	\[
	F_{\gamma',\alpha',k}(X\cup\{y\})-F_{\gamma',\alpha',k}(X)=F_{\gamma',\alpha',k}(\{y\})\leq F_{\gamma',\alpha',k}(\{b_{1}\})=\frac{1}{k}
	\]
	and, for any $\alpha\geq1$, we have
	\[
	\frac{F_{\gamma',\alpha',k}(X\cup Y)-\alpha F_{\gamma',\alpha',k}(X)}{|Y|}=\frac{F_{\gamma',\alpha',k}(B)}{k}=\frac{1}{k\gamma'}>\frac{1}{k}
	\]
	because $\gamma'<1$. Thus, $F_{\gamma',\alpha',k}$ is not $\alpha$-augmentable
	for any $\alpha\geq1$.
	
	For the second part, let $\gamma'\in(0,1]$, $\alpha'\geq\gamma'$
	and $k\in\N$ with $k>\alpha'$. Furthermore, let $X=A=\greedyset_{k}$
	and $Y=B$. For $y\in Y$, we have
	\[
	F_{\gamma',\alpha',k}(X\cup\{y\})-F_{\gamma',\alpha',k}(X)=\begin{cases}
		\frac{1}{k}\bigl(1-\alpha'\sum_{i=1}^{k}\xi_{i}\bigr)=\frac{1}{k}\bigl(\frac{k-\alpha'}{k}\bigr)^{k} & \textrm{if }y=b_{1},\\
		0, & \textrm{else},
	\end{cases}
	\]
	and
	\[
	F_{\gamma',\alpha',k}(X\cup Y)-F_{\gamma',\alpha',k}(X)=F_{\gamma',\alpha',k}(B)-F_{\gamma',\alpha',k}(A)=\frac{1}{\gamma'}-\frac{1}{\alpha'}\bigl(1-\bigl(\frac{k-\alpha'}{k}\bigr)^{k}\bigr).
	\]
	For $k\rightarrow\infty$, the (weak) submodularity ratio gets arbitrarily
	close to 0 because
	\[
	\lim_{k\rightarrow\infty}\frac{\sum_{y\in Y}\bigl(F_{\gamma',\alpha',k}(X\cup\{y\})-F_{\gamma',\alpha',k}(X)\bigr)}{F_{\gamma',\alpha',k}(X\cup Y)-F_{\gamma',\alpha',k}(X)}=\lim_{k\rightarrow\infty}\frac{\frac{1}{k}\bigl(\frac{k-\alpha'}{k}\bigr)^{k}}{\frac{1}{\gamma'}-\frac{1}{\alpha'}\bigl(1-\bigl(\frac{k-\alpha'}{k}\bigr)^{k}\bigr)}=0,
	\]
	i.e., for $k=k'$ large enough, $F_{\gamma',\alpha',k'}\notin\srfuncs_{\gamma}$.
	It remains to show that $F_{\gamma',\alpha',k'}\nsubseteq\rqfuncs_{q}$.
	If $F_{\gamma',\alpha',k'}\in\rqfuncs_{q}$ would hold, then there
	would be some independence system with weight function $w$ such that
	$F_{\gamma',\alpha',k'}$ was the associated weighted rank function.
	The fact that $F_{\gamma',\alpha',k'}(\{b_{2}\})=0$ implies that~$b_{2}$
	must have weight 0 or~$\{b_{2}\}$ is not independent, and the fact
	that $F_{\gamma',\alpha',k'}(\{b_{1},b_{2}\})-F_{\gamma',\alpha',k'}(\{b_{1}\})=\frac{h(2)-h(1)}{k}>0$
	implies that $b_{2}$ must have a weight greater 0 and that $\{b_{2}\}$
	has to be independent, which contradict each other. Thus,~$F_{\gamma',\alpha',k'}$
	cannot be modelled as the weighted rank function of an independence
	system, i.e., $F_{\gamma',\alpha',k'}\notin\rqfuncs_{q}$.
\end{proof}

Finally, we can extend Proposition~\ref{prop:separating_alpha_aug_lite}
to all $\alpha\geq1$ by combining the fact that, by Proposition~\ref{prop:F_1,alph,k_is_alph-aug},
for $\alpha\geq1$, $\mbox{\ensuremath{\{F_{1,\alpha,k}\mid k\in\N,k>\alpha\}\subseteq\augfuncs_{\alpha}}}$
and the fact that, by Proposition~\ref{prop:F_separates_gam-alph-aug},
for every $\gamma,q\in(0,1]$, $\{F_{1,\alpha,k}\mid k\in\N,k>\alpha\}\nsubseteq\srfuncs_{\gamma}\cup\rqfuncs_{q}$.
\begin{prop}
	For every $\gamma,q\in(0,1]$, $\alpha\geq1$, it holds that $\augfuncs_{\alpha}\nsubseteq(\srfuncs_{\gamma}\cup\rqfuncs_{q})$.\label{prop:separating_alpha_aug}
\end{prop}

\subsection{$\gamma$-$\alpha$-Augmentability on Independence Systems\label{subsec:IS}}

To tightly capture the class $\rqfuncs_{q}$ of weighted rank functions
on independence systems, we show a stronger bound for the approximation
ratio of the greedy algorithm on monotone, (weakly) $\gamma$-$\alpha$-augmentable
functions. In particular, it was already shown in~\cite{BernsteinDisserGrossHimburg/20}
that the objective function of $\alpha\textsc{-Dimensional Matching}$
is (exactly) $\alpha$-augmentable, while the greedy algorithm yields
an approximation ratio of $\alpha$, which beats the upper bound of
$\alpha\cdot\frac{\e^{\alpha}}{\e^{\alpha}-1}$ for this case. We
show that this can be explained by the fact that $\alpha\textsc{-Dimensional Matching}$
can be represented via a weighted rank function over an indepencence
system. We first show the upper bound of Theorem~\ref{thm:IS_bound}.
\begin{prop}
	\label{prop:upper-bound_gamma-alpha-aug_ind-sys}Let $\rqfuncs_{\mathrm{IS}}:=\bigcup_{q\in(0,1]}\rqfuncs_{q}$
	be the set of weighted rank functions on some independence system.
	The approximation ratio of the greedy algorithm on the class $\gafuncs_{\gamma,\alpha}\cap\rqfuncs_{\mathrm{IS}}$
	is at most~$\frac{\alpha}{\gamma}$, for every $\gamma\in(0,1]$
	and $\alpha\geq\gamma$.
\end{prop}

\begin{proof}
	Let $f\in\gafuncs_{\gamma,\alpha}\cap\rqfuncs_{\mathrm{IS}}$, and
	let $w\colon\U\rightarrow\R_{\geq0}$ be the weight function that
	induces $f$. We use induction over $k$. For $k=0$, the statement
	holds obviously. Now suppose, the statement holds for some $k\in\oneto{|\U|-1}$.
	If $f(\greedyset_{k})\geq\frac{\gamma}{\alpha}f(\optset_{k+1}),$
	then, by monotonicity of $f$, we have $f(\greedyset_{k+1})\geq f(\greedyset_{k})\geq\frac{\gamma}{\alpha}f(\optset_{k+1})$.
	Otherwise, the weak $\gamma$-$\alpha$-augmentability of $f$ guarantees
	the existence of $x\in\optset_{k+1}$ with
	\[
	f(\greedyset_{k}\cup\{x\})-f(\greedyset_{k})\geq\frac{\gamma f(\greedyset_{k}\cup\optset_{k+1})-\alpha f(\greedyset_{k})}{|\optset_{k+1}|}\geq\frac{\gamma f(\optset_{k+1})-\alpha f(\greedyset_{k})}{k+1}>0.
	\]
	By Lemma \ref{lem:adding_element_to_greedy_set}, this is equivalent
	to $f(\greedyset_{k}\cup\{x\})=f(\greedyset_{k})+w(x).$ We conclude\belowdisplayskip=-12pt
	\begin{eqnarray*}
		f(\greedyset_{k+1}) & \geq & f(\greedyset_{k}\cup\{x\})\\
		& = & f(\greedyset_{k})+w(x)\\
		& \overset{\textrm{ind}}{\geq} & \frac{\gamma}{\alpha}f(\optset_{k})+w(x)\\
		& \geq & \frac{\gamma}{\alpha}f(\optset_{k+1}\setminus\{x\})+w(x)\\
		& \overset{\alpha\geq\gamma}{\geq} & \frac{\gamma}{\alpha}f(\optset_{k+1}\setminus\{x\})+\frac{\gamma}{\alpha}w(x)\\
		& \geq & \frac{\gamma}{\alpha}f(\optset_{k+1}).
	\end{eqnarray*}
\end{proof}

The lower bound of Theorem~\ref{thm:IS_bound} follows directly from
the well-known tight bound of~$1/q$ for $\rqfuncs_{q}$.
\begin{prop}
	\label{prop:lower-bound_gamma-alpha-aug_ind-sys}Let $\rqfuncs_{\mathrm{IS}}:=\bigcup_{q\in(0,1]}\rqfuncs_{q}$
	be the set of weighted rank functions on some independence system.
	The approximation ratio of the greedy algorithm on the class $\gafuncs_{\gamma,\alpha}\cap\rqfuncs_{\mathrm{IS}}$
	is at least $\frac{\alpha}{\gamma}$, for every $\gamma\in(0,1]$
	and $\alpha\geq\gamma$.
\end{prop}

\begin{proof}
	Let $\gamma\in(0,1]$, $\alpha\in\R$ and $q\in[\frac{\gamma}{\alpha},1]\cap\Q$.
	In~\cite{Jenkyns1976} it was shown that the the approximation ratio
	of the greedy algorithm on the set $\rqfuncs_{q}$ is exactly $1/q$.
	By definition of $\rqfuncs_{\mathrm{IS}}$, we have $\rqfuncs_{q}\subseteq\rqfuncs_{\mathrm{IS}}$,
	and, by Proposition~\ref{prop:gamma-alpha_unifies_alpha_and_submod-ratio},
	$\rqfuncs_{q}\subseteq\gafuncs_{\gamma,\gamma/q}\subseteq\gafuncs_{\gamma,\alpha}$
	holds, where we use the fact that $\frac{\gamma}{q}\leq\frac{\gamma}{\gamma/\alpha}=\alpha$.
	Thus, we can conclude that the approximation ratio of the greedy algorithm
	on the class $\smash{\gafuncs_{\gamma,\alpha}\cap\rqfuncs_{\mathrm{IS}}}$
	is at least~$1/q$, and since $q$ can be chosen arbitrarily close
	to $\frac{\gamma}{\alpha}$, the statement follows.
\end{proof}

It can be shown that the lower bound of Proposition~\ref{prop:lower-bound_gamma-alpha-aug_ind-sys}
already holds for \mbox{$\gamma$-$\alpha$-augmentable} functions, i.e.,
in the non-weak subclass of $\gafuncs_{\gamma,\alpha}$. It follows
that the tight bound of Theorem~\ref{thm:IS_bound} carries over
to this, in some sense more natural, class of functions. Since every
$\alpha$-augmentable function is 1-$\alpha$-augmentable, and vice-versa,
we additionally obtain the following. Note that this tightly captures
the perfomance of the greedy algorithm for the $\alpha\textsc{-Dimensional Matching}$
problem, which can be represented as the maximization of an $\alpha$-augmentable
weighted rank function over an independence system~\cite{BernsteinDisserGrossHimburg/20}.
\begin{cor}
	The approximation ratio of the greedy algorithm on the class $\augfuncs_{\alpha}\cap\rqfuncs_{\mathrm{IS}}$,
	with $\alpha\geq1$, is exactly~$\alpha$.
\end{cor}

\section{Outlook}

The vision guiding our work is to precisely characterize the set of
cardinality-constrained maximization problems for which the greedy
algorithm yields an approximation, and to tightly bound the corresponding
approximation ratio.

In this paper, we have made progress towards this goal by unifying
and generalizing important classes of greedily approximable maximization
problems, and by providing tight bounds on the approximation ratio
for the resulting generalized class of problems. While this brings
us closer to a full characterization, there are still settings that
are not captured by (weak) $\gamma$-$\alpha$-augmentability.
\begin{prop}
	\label{prop:greedy_optimal_not_gam-alph-aug}For $\gamma\in(0,1]$
	and $\alpha\geq\gamma$, there exists a monotone function $f^{\gamma,\alpha}$
	that is not weakly $\gamma$-$\alpha$-augmentable, and for which
	the greedy algorithm computes an optimum solution.
\end{prop}

\begin{proof}
	Let $U$ be any ground set of size $|U|>\frac{1}{\gamma}$ and consider
	the objective function $f^{\gamma,\alpha}\colon2^{\U}\rightarrow\R_{\geq0}$
	with
	\[
	f^{\gamma,\alpha}(X)=|X|^{2}.
	\]
	For all $X,Y\subseteq U$ with $|Y|>\frac{1}{\gamma}(2|X|+1+\alpha|X|^{2})$
	$(*)$, e.g., $X=\emptyset$ and $|Y|=\lfloor\frac{1}{\gamma}\rfloor+1$,
	we have
	\begin{eqnarray*}
		|Y|(f^{\gamma,\alpha}(X\cup\{y\})-f^{\gamma,\alpha}(X)) & = & |Y|(2|X|+1)\\
		& \overset{(*)}{<} & \gamma|Y|^{2}-\alpha|Y||X|^{2}\\
		& \leq & \gamma|X\cup Y|^{2}-\alpha|X|^{2}\\
		& = & \gamma f^{\gamma,\alpha}(X\cup Y)-\alpha f^{\gamma,\alpha}(X),
	\end{eqnarray*}
	i.e., $f^{\gamma,\alpha}$ is not weakly $\gamma$-$\alpha$-augmentable.
	Yet, picking elements in any order is obviously optimal. Thus, there
	exists a problem that is not $\gamma$-$\alpha$-augmentable, but
	where the greedy algorithm performs optimally.
\end{proof}
\begin{rem}
	Objective functions as in the proof of Proposition \ref{prop:greedy_optimal_not_gam-alph-aug}
	arise for example in the context of incremental maximum flows on a
	complete bipartite graph $G=(U\cup V,E)$ where we want to incrementally
	grow subsets of $U$ and of $V$ such that the flow from one of the
	subsets to the other (i.e., the cut size) is maximized.
\end{rem}

We leave it as an open problem to find a natural generalization of
weak $\gamma$-$\alpha$-augmentability that captures a larger set
of greedily approximable objectives. The challenge is to find a meaningful
generalization in terms of a natural definition that does not directly
depend on the behavior of the greedy algorithm, but rather enforces
some structural property of the objective function. In that sense,
the dependency of weak $\gamma$-$\alpha$-augmentability on the greedy
solutions $\greedyset_{0},\dots,\greedyset_{\lastk}$ is a significant
flaw. Note that we needed to introduce this dependency in order to
encompass settings with bounded (weak) submodularity ratios, since
the definition of the latter depends on the greedy solutions as well.
Importantly, our upper bound on the approximation ratio of the greedy
algorithm carries over to the stronger notion of $\gamma$-$\alpha$-augmentability
that requires the defining property to hold for \textit{all} sets $X$,
and not just the greedy solutions. Our tight lower bound does not
immediately translate to this, more restrictive, definition, and it
remains an open problem to construct a tight lower bound in this setting
as well.

\bibliographystyle{siamplain}
\bibliography{bibliography}
\end{document}